\documentclass[11pt,letterpaper]{article}
\usepackage[margin=24.5mm]{geometry}

\usepackage{amsmath, amsthm, amsfonts}

\usepackage{thmtools}
\usepackage{xspace}

\usepackage{algpseudocode}

\usepackage{xcolor}
\usepackage{tikz,pgf}
\usetikzlibrary{arrows.meta, positioning}

\usepackage{hyperref}
\usepackage{orcidlink}

\definecolor{darkgreen}{rgb}{0,0.5,0}
\definecolor{darkred}{rgb}{0.65,0,0}
\hypersetup{
    colorlinks=true,
    linkcolor=darkred,
    citecolor=darkgreen,
    filecolor=black,
    urlcolor=[rgb]{0,0.1,0.5},
    pdftitle={Distributed (Delta+1)-Coloring in Graphs of Bounded Neighborhood Independence},
    pdfauthor={Marc Fuchs, Fabian Kuhn},
    pdfkeywords = {distributed computing, distributed graph algorithms, graph coloring, list coloring, defective coloring},
}

\usepackage{cleveref}

\newtheorem{theorem}{Theorem}[section]
\newtheorem{lemma}[theorem]{Lemma}

\newtheorem{corollary}[theorem]{Corollary}

\newtheorem{definition}{Definition}[section]

\newcommand{\calA}{\ensuremath{\mathcal{A}}}

\newcommand{\calC}{\ensuremath{\mathcal{C}}}

\newcommand{\CONGEST}{\ensuremath{\mathsf{CONGEST}}\xspace}
\newcommand{\DCLOCAL}{\ensuremath{\mathsf{DC}\text{-}\LOCAL}\xspace}

\newcommand{\LOCAL}{\ensuremath{\mathsf{LOCAL}}\xspace}

\newcommand{\eps}{\varepsilon}
\renewcommand{\epsilon}{\varepsilon}
\newcommand{\poly}{\operatorname{\text{{\rmfamily poly}}}}

\newcommand{\set}[1]{\left\{#1\right\}}

\newcommand{\hide}[1]{}

\renewcommand{\phi}{\varphi}

\begin{document}

\title{Distributed $(\Delta+1)$-Coloring in Graphs \\ of Bounded Neighborhood Independence\footnote{This work was partially supported by the German Research Foundation (DFG) under the project number 491819048.}}

\date{}

\author{
   Marc Fuchs \orcidlink{0000-0003-2272-4483} \\
   \small{University of Freiburg} \\
   \small{marc.fuchs@cs.uni-freiburg.de}
   \and
   Fabian Kuhn \orcidlink{0000-0002-1025-5037} \\
   \small{University of Freiburg} \\
   \small{kuhn@cs.uni-freiburg.de}
}

\maketitle

\begin{abstract}
  The distributed coloring problem is arguably one of the key problems studied in the area of distributed graph algorithms. The most standard variant of the problem asks for a proper vertex coloring of a graph with $\Delta+1$ colors, where $\Delta$ is the maximum degree of the graph. Despite an immense amount of work on distributed coloring problems in the distributed setting, determining the deterministic complexity of $(\Delta+1)$-coloring in the standard message passing model remains one of the most important open questions of the area. In the \LOCAL model, it is known that $(\Delta+1)$-coloring requires $\Omega(\log^* n)$ rounds even in paths and rings (i.e., when $\Delta=2$). For general graphs, the problem is known to be solvable in $\tilde{O}\big(\log^{5/3}n\big)$ rounds and in $O(\sqrt{\Delta\log\Delta} + \log^* n)$ rounds when expressing the complexity as a function of $\Delta$ and with an optimal dependency on $n$. In the present paper, we aim to improve our understanding of the deterministic complexity of $(\Delta+1)$-coloring as a function of $\Delta$ in a special family of graphs for which significantly faster algorithms are already known.

  The neighborhood independence $\theta$ of a graph is the maximum number of pairwise non-adjacent neighbors of some node of the graph. Notable examples of graphs of bounded neighborhood independence are line graphs of graphs and bounded-rank hypergraphs. It is known that the $(2\Delta-1)$-edge coloring problem and therefore the $(\Delta+1)$-coloring problem in line graphs of graphs can be solved in $O\big(\log^{12}\Delta+\log^* n\big)$ rounds. In general, in graphs of neighborhood independence $\theta=O(1)$, it is known that $(\Delta+1)$-coloring can be solved in $2^{O(\sqrt{\log\Delta})}+O(\log^* n)$ rounds.  In the present paper, we significantly improve the latter result, and we show that in graphs of bounded neighborhood independence $\theta$, a $(\Delta+1)$-coloring can be computed in $(\log\Delta)^{O\left(1 + \frac{\log\log\Delta}{ \log\log\log\Delta}\right)}+O(\log^* n)$ rounds and thus in quasipolylogarithmic time in $\Delta$ as long as $\theta$ is constant. Our algorithm can be seen as a generalization of an existing similar, but slightly weaker result for $(2\Delta-1)$-edge coloring. We also show that the approach that leads to this polylogarithmic in $\Delta$ algorithm for $(2\Delta-1)$-edge coloring already fails for edge colorings of hypergraphs of rank at least $3$. At the core of the fast edge coloring algorithm is an algorithm to divide the edges of a graph into two parts so that up to a multiplicative error of $1+o(1)$, the maximum degree of the line graph induced by each part is at most half the maximum degree of the original line graph. We show that computing such a bipartition of the edges of the line graph of a hypergraph of rank at least $3$ requires time logarithmic in $n$.
\end{abstract}

\newpage

\section{Introduction and Related Work}
\label{sec:intro}

The distributed coloring problem asks for a proper vertex coloring of some graph that is provided in a distributed way. In the most standard form, the task is to color the nodes of some $n$-node graph $G=(V,E)$, which also defines the network topology. The nodes of $G$ have unique $O(\log n)$-bit identifiers and they can exchange messages over the edges of $G$ in synchronous rounds. If the messages can be of arbitrary size, this setting is known as the \LOCAL model~\cite{Linial1987} and if the message size is restricted to $O(\log n)$ bits, it is known as the \CONGEST model~\cite{Peleg2000}.\footnote{Throughout this paper, we focus on the \LOCAL model. We however remark that with relatively small additional work, the algorithm that we present in the paper can also be implemented in the \CONGEST model.} The number of allowed colors is usually some function of the maximum degree $\Delta$. Of particular interest is the problem of coloring with $\Delta+1$ colors, i.e., the number of colors used by a sequential greedy algorithm.

The distributed coloring problem is possibly the most intensively studied problem in the area of distributed graph algorithms. The most standard version of the problem asks for a vertex coloring of the network graph $G$ with $\Delta+1$ colors, where $\Delta$ denotes the maximum degree of $G$. Despite the vast amount of previous work on the problem and despite astonishing progress over the last years, the gap between the best known lower and upper bounds on the complexity of $(\Delta+1)$-coloring in \LOCAL and \CONGEST is still very large. The only known lower bound is an $\Omega(\log^* n)$ lower bound (in \LOCAL and \CONGEST), which holds in paths and cycles (i.e., even if $\Delta=2$) and was proven for deterministic algorithms by Linial~\cite{Linial1987} and for randomized algorithms by Naor~\cite{Naor1991}. There is a long series of work that tried to optimize the complexity of the problem as a function of $n$~\cite{Awerbuch89,bamberger2020coloring,Barenboim2013,BarenboimE11,Barenboim2016,ChangLP18,GhaffariGrunau23,GhaffariGrunauFOCS24,GGR2020,GhaffariKuhn21,HalldorssonKMT21,HarrisSS18,johansson99,Kuhn20,Linial1987,Luby1986,panconesi96decomposition,Rozhon2020,SchneiderW10symmetry}. The best known upper bounds in \LOCAL are $\tilde{O}(\log^{5/3} n)$ for deterministic~\cite{GhaffariGrunauFOCS24} and $\tilde{O}(\log^{5/3}\log n)$ for randomized algorithms~\cite{ChangLP18,GhaffariGrunauFOCS24}.\footnote{The notation $\tilde{O}(\cdot)$ hides polylogarithmic factors in the argument, i.e., $\tilde{O}(x)=O(x\cdot\poly\log x)$.} In the \CONGEST model, the best known algorithms have a round complexity of $O(\log^2\Delta\cdot\log n)$ for deterministic algorithms~\cite{GhaffariKuhn21} and $O(\log^3\log n)$ for randomized algorithms~\cite{HalldorssonKMT21}. Apart from work that tries to optimize the complexity primarily as a function of $n$, there is also a long line of research that tries to optimize the complexity of $(\Delta+1)$-coloring problems as a function of $\Delta$, while keeping the dependency on $n$ of order $O(\log^* n)$ and thus asymptotically optimal~\cite{barenboim16sublinear,barenboim11,Barenboim2013,BarenboimEG18,barenboim14distributed,fraigniaud16local,FK23,goldberg88,Kuhn20,KuhnW06,Kuhn2009,MausT20,SzegedyV93}. In this case, it is known that randomization does not help~\cite{chang16exponential} and people have therefore only studied deterministic algorithms. The fastest such $(\Delta+1)$-coloring algorithm in the \LOCAL model requires $O(\sqrt{\Delta\log\Delta} + \log^* n)$ rounds~\cite{fraigniaud16local,BarenboimEG18,MausT20}. The best \CONGEST algorithm is slightly slower and uses $O(\sqrt{\Delta}\cdot\log^2\Delta\cdot\log^6\log\Delta + \log^* n)$ rounds~\cite{FK23}.

Determining the true \textbf{complexity of distributed \boldmath$(\Delta+1)$-coloring} and thus closing the large gaps between the above upper bounds and the $\Omega(\log^* n)$ lower bound of \cite{Linial1987,Naor1991} is arguably \textbf{one of the most important open problems} in the area of distributed graph algorithms. Note that for the closely related problem of computing a maximal independent set (MIS), the situation is very different. The best known deterministic upper bounds for computing an MIS are  $O(\Delta+\log^* n)$ (in \LOCAL and \CONGEST)~\cite{barenboim14distributed}, $\tilde{O}(\log^{5/3} n)$ in \LOCAL~\cite{GhaffariGrunauFOCS24}, and $O(\log^2\Delta\cdot \log n)$ in \CONGEST~\cite{Localrounding23} and the best randomized upper bounds are $\tilde{O}(\log\Delta + \log^{5/3}\log n)$ in \LOCAL~\cite{ghaffari16improved,GhaffariGrunauFOCS24}, and $O(\log\Delta + \log^3\log n)$ in \CONGEST~\cite{ghaffari16improved,Localrounding23}. All of those upper bounds are at least as large as the corresponding upper bounds for $(\Delta+1)$-coloring.\footnote{This is no coincidence as there is a general locality-preserving reduction from the distributed $(\Delta+1)$-coloring problem to the distributed MIS problem~\cite{Luby1986,Linial1987}.} However, for the problem of computing an MIS, there are also strong lower bounds known, which show that the known upper bounds are either optimal or relatively close to being optimal. In \cite{KMW04}, it was shown that even with randomization, computing an MIS requires at least $\Omega\Big(\min\set{\frac{\log\Delta}{\log\log\Delta}, \log_\Delta n}\Big)$ rounds (implying an $\Omega\Big(\sqrt{\frac{\log n}{\log\log n}}\Big)$ lower bound as a function of $n$). Further, in \cite{Balliu2019}, it was proven that computing an MIS deterministically requires $\Omega\big(\min\set{\Delta, \log_\Delta n}\big)$, leading to an $\Omega\big(\frac{\log n}{\log\log n}\big)$ lower bound as a function of $n$. The lower bound of \cite{Balliu2019} in particular implies that when restricting the $n$-dependency to $O(\log^* n)$, the $O(\Delta+\log^* n)$ bound of \cite{barenboim14distributed} is optimal. In \cite{hideandseek,Balliu0KO22} it was further shown that all the MIS lower bounds of \cite{KMW04,Balliu2019} even hold if the network graph is a tree.

When trying to improve our understanding of the distributed $(\Delta+1)$-coloring problem, we believe that the $\Delta$-dependency of the problem complexity is currently the most interesting case. First note that we can concentrate on deterministic algorithms as it is well understood that the current randomized upper bounds can only be improved by also improving the deterministic upper bounds~\cite{chang16exponential,ChangLP18}. Further, one can use techniques similar to the ones used in \cite{chang16exponential} to show that any deterministic \LOCAL algorithm with round complexity $f(n)=o(\log n / \log\log n)$ would directly imply a deterministic \LOCAL algorithm with round complexity $\Delta^{o(1)}\cdot \log^* n$.\footnote{To see this, take any value $t$ such that $f\big(\Delta^{2t}\big)<t/2$. In $t\cdot \log^* n$ rounds, one can color $G^t$ with $\Delta^{2t}$ colors by using an algorithm of Linial~\cite{Linial1987}. As $f\big(\Delta^{2t}\big)<t/2$, one can then use those colors as unique IDs to run the given $f(n)$-round algorithm in $f\big(\Delta^{2t}\big)$ rounds. If $f(n)<\eps\cdot\log n /\log\log n$, we can satisfy $f\big(\Delta^{2t}\big)<t/2$ for some $t\leq \Delta^{O(\eps)}$.} Hence, unless we can improve the existing $O(\sqrt{\Delta\log\Delta} + \log^* n)$ round algorithm of \cite{fraigniaud16local,BarenboimEG18,MausT20} very significantly, there cannot be any deterministic $(\Delta+1)$-coloring algorithm with a complexity $o(\log n / \log\log n)$, which would mean that the existing $\tilde{O}(\log^{5/3} n)$-round algorithm of \cite{GhaffariGrunauFOCS24} is already quite close to optimal.

In the present paper, we make progress on the $\Delta$-dependency of $(\Delta+1)$-coloring for a special family of graphs for which faster $(\Delta+1)$-coloring algorithms are already known: graphs of \emph{bounded neighborhood independence}. The neighborhood independence $\theta$ of a graph $G=(V,E)$ is the maximum independent set size of any one-hop neighborhood of $G$, that is, the maximum independence number of any of the induced subgraphs $G[N(v)]$ for $v\in V$. Equivalently, graphs of neighborhood independence $\theta$ are the graphs that do not have $K_{1,\theta+1}$ as an induced subgraph and the family of graphs with neighborhood independence $\theta=2$ is sometimes also known as \emph{claw-free graphs}~\cite{FaudreeFlandrinRyjacek1997}. The family of bounded neighborhood independence graphs is the family of graphs for which $\theta=O(1)$. Important examples of graphs of bounded neighborhood independence are the line graphs of hypergraphs of bounded rank. Because in a line graph of rank $r$, the neighboring edges of each hyperedge can be partitioned into $r$ cliques (in the line graph), the neighborhood independence of the line graph can be at most $r$. Hence, the problem of $(2\Delta-1)$-edge coloring is a special case of $(\Delta+1)$-coloring graphs of bounded neighborhood independence. More generally, the problem of $(r\Delta-(r-1))$-edge coloring hypergraphs of rank $r$ is a special case of $(\Delta+1)$-coloring graphs of neighborhood independence $\theta=r$. It was proven in \cite{Kuhn20} that $(\Delta+1)$-coloring graph of neighborhood independence $\theta$ can be done in time $2^{O(\sqrt{\log\theta\cdot \log\Delta})}+O(\log^* n)$. The special case of $(2\Delta-1)$-edge coloring (of graphs) was later solved in time $(\log\Delta)^{O(\log\log\Delta)}+O(\log^* n)$ in \cite{BalliuKO20} and in time $O(\log^{12}\Delta + \log^* n)$ in \cite{BalliuBKO22}. The algorithms of \cite{BalliuKO20,BalliuBKO22} were specifically designed for the $(2\Delta-1)$-edge coloring problem and they cannot directly be applied to even coloring the edges of hypergraphs of rank at least $3$. As our main result, we provide an algorithm that parameterizes the complexity of $(\Delta+1)$-coloring as a function of the neighborhood independence $\theta$ and the maximum degree $\Delta$ and that for $\theta=O(1)$ is slightly faster than the $(2\Delta-1)$-edge coloring algorithm  of \cite{BalliuKO20}.

\begin{restatable}{theorem}{restateMain}\label{thm:mainNeighborhood}
  There is a deterministic \LOCAL algorithm with round complexity
  \begin{align*}
      \min \left\{(\theta\log\Delta)^{O\left(\frac{\log\log \Delta}{\max \{1, \log\log\log \Delta - \log \log \theta \}}\right)}, O( \theta^2 \cdot \Delta^{\frac{1}{4}} \cdot \log^{7} \Delta)  \right\} + O(\log^* n)
  \end{align*}
  to solve $(\Delta+1)$-coloring in any $n$-node graph $G$ of maximum degree $\Delta$ and neighborhood independence $\theta$.
\end{restatable}

In the above theorem and throughout the paper, we use the convention $\log x:=\max\{1,\log_2 x\}$; iterated logarithms are obtained by applying this convention repeatedly. Thus all displayed logarithmic bounds are also well-defined for small parameter values. A direct implication of \Cref{thm:mainNeighborhood} is that graphs with bounded neighborhood independence $\theta = O(1)$ can be properly colored with $\Delta+1$ colors in $(\log\Delta)^{O\left(\frac{\log \log \Delta}{\log \log \log \Delta}\right)} + O(\log^* n)$ communication rounds. This, for instance, implies that an $(r(\Delta-1)+1)$-edge coloring of hypergraphs of rank $r=O(1)$ can be computed in time $(\log\Delta)^{O\left(\frac{\log \log \Delta}{\log \log \log \Delta}\right)} + O(\log^* n)$.

Another implication is that we get faster $(\Delta+1)$-coloring algorithms and, more generally, faster $(\mathit{degree}+1)$-list coloring algorithms in \LOCAL even if the neighborhood independence is moderately large. Specifically if $\theta = \tilde{o}(\Delta^{1/8})$, we obtain an algorithm with round complexity $o(\sqrt{\Delta}) + O(\log^*n)$, which is faster than the currently best such algorithm for general graphs, which has a round complexity of $O(\sqrt{\Delta\log\Delta}+\log^* n)$~\cite{fraigniaud16local,BarenboimEG18,MausT20}.

Small neighborhood independence might seem to be a quite restrictive graph parameter. We however emphasize that it generally does not make local graph problems easy. The MIS lower bounds of \cite{KMW04,Balliu2019} are actually proven for the simpler maximal matching problem and thus for the MIS problem on graphs of neighborhood independence $2$. We further remark that while our algorithm is slower than the $O(\log^{12}\Delta+\log^* n)$-round algorithm of \cite{BalliuBKO22} for the $(2\Delta-1)$-edge coloring problem (on graphs), the generalization to general bounded neighborhood independence and even to just edge-coloring bounded-rank hypergraphs seems significant. At the core of the edge coloring algorithm of \cite{BalliuBKO22} is an efficient defective $2$-edge-coloring algorithm with very tight parameters. Specifically, given an arbitrary graph $G=(V,E)$, let the edge degree $\deg(e)$ of an edge $e=\set{u,v}\in E$ be $\deg(u)+\deg(v)-2$. In \cite{BalliuBKO22}, it is in particular shown that one can efficiently color the edges of a graph with two colors such that for each edge $e\in E$, at most $(1/2+o(1))\cdot\deg(e)$ of its neighboring edges have the same color as $e$. This algorithm exploits the specific properties of the line graphs of graphs, and we do not see a way to generalize the algorithm even to hypergraphs of rank $3$. We formally show that as long as the algorithm is based on defective edge coloring with $2$ colors to recursively decompose the graph, the strategy cannot be extended to $r$-uniform hypergraphs for any $r\geq 3$. More specifically, we prove the following theorem by adapting an existing similar lowerbound for defective $2$-vertex-coloring of regular trees from \cite{BalliuHLOS19}.

\begin{restatable}{theorem}{restateLowerBound}\label{thm:lowerbound}
 For all integers $d,r\geq 3$, $d,r\in O(1)$, in $r$-uniform hypergraphs of maximum degree $d$, any deterministic \LOCAL algorithm needs at least $\Omega(\log n)$ rounds to compute a $(r(d-1)-d)$-defective $2$-edge-coloring.
\end{restatable}
Note that an edge coloring is called $\delta$-defective if every edge has at most $\delta$ neighboring edges of the same color. Also note that the maximum edge degree of a $r$-uniform hypergraph of degree $\leq d$ is at most $r(d-1)$. The above bound thus implies that for $r\geq 3$, and $r,d=O(1)$, one needs $\Omega(\log n)$ rounds to compute a $\big(1-\frac{1}{d-1}\big)\cdot(r-1)/r\cdot \Delta_E$-defective $2$-edge-coloring in rank $r$-hypergraphs with maximum edge degree $\Delta_E$. Further, since line graphs of rank-$r$ hypergraphs have neighborhood independence $\theta=r$, the lower bound also implies that in $\Delta$-regular graphs with neighborhood independence $\theta\geq 3$ (for $\Delta=O(1)$ and $\theta$ divides $\Delta$), one needs $\Omega(\log n)$ rounds to deterministically compute a defective $2$-coloring with defect $\leq\big(1-\frac{\theta}{\Delta}\big)\cdot\frac{\theta-1}{\theta}$.

\subsection{Organization of the paper}
\label{sec:organization}
In \Cref{sec:overview}, we give a high-level description of all our technical steps that are needed to prove \Cref{thm:mainNeighborhood,thm:lowerbound}. The algorithm that leads to the upper bound of \Cref{thm:mainNeighborhood} is a recursive algorithm. The recursive subproblems are defined in \Cref{sec:recsubproblems}. The recursive structure consists of two main steps,  \emph{slack generation}  and \emph{color space reduction}, which are described in \Cref{sec:intro:slackgeneration,sec:intro:colspacereduction}, respectively. At the core of the color space reduction part is an algorithm to efficiently compute \emph{list defective colorings} in graphs with small neighborhood independence by using an efficient algorithm to compute good quality \emph{list arbdefective colorings} (definitions, see below). \Cref{sec:overview} is concluded with \Cref{sec:intro:lower}, which describes the high-level ideas of the lower bound result given by \Cref{thm:lowerbound}. The technical details of the paper appear in three sections. \Cref{sec:arbtodef,sec:MainSection} contain the main technical contributions of the paper, the details on how to get from list arbdefective to list defective colorings and the recursive structure and analysis of our main algorithm. The formal proof of the lower bound appears in \Cref{sec:lowerbound}.

\section{Technical Overview}
\label{sec:overview}

We next discuss the main technical ideas that are necessary to prove \Cref{thm:mainNeighborhood,thm:lowerbound}. As \Cref{thm:lowerbound} is an adaptation of a similar lower bound in \cite{BalliuHLOS19}, we primarily focus on \Cref{thm:mainNeighborhood}. Before discussing the actual algorithm leading to our main theorem, we have to introduce different variants of \emph{defective colorings}.

\subsection{Defective Colorings}\label{sec:defective}

Defective colorings are a key tool that is used in essentially all modern deterministic distributed coloring algorithms. In its most standard form, for an integer parameter $d\geq 0$, a $c$-coloring of the nodes of a graph $G=(V,E)$ is called $d$-defective if each node has at most $d$ neighbors of the same color. Defective colorings have been used to decompose a given graph into a small number of subgraphs of smaller degree that can then be colored more efficiently either in parallel (e.g., \cite{BalliuKO20,BalliuBKO22,BarenboimE09,BarenboimE10,barenboim11,Kuhn20,GhaffariKuhn21}) or sequentially (e.g., \cite{BarenboimE09,BarenboimE10,barenboim11,Barenboim2016,fraigniaud16local,Kuhn20,GhaffariKuhn21}). Although defective coloring might seem to be a simple relaxation of proper graph coloring, it is algorithmically quite different. While it is known that every graph has a $d$-defective coloring with $\lceil\frac{\Delta+1}{d+1}\rceil$ colors~\cite{lovasz66}, for general values of $d$, there is no $f(\Delta)\cdot\log^* n$-round distributed algorithm known to compute a $d$-defective coloring even with $O(\Delta/d)$ colors. The only such algorithm known is an $O(\log^* n)$-round algorithm to compute a $d$-defective coloring with $O\big(\big(\frac{\Delta}{d+1}\big)^2\big)$ colors~\cite{Kuhn2009}. Since we use a variant of this algorithm as a tool in our algorithm, we state it formally.

\begin{lemma}[\cite{Kuhn2009, KawarabayashiS18}]
    \label{lemm:defColorBlackBox}
    For any $1 \geq \alpha > 0$, there exists a deterministic \CONGEST algorithm to compute a defective coloring of a given graph $G$ with $O(1/\alpha^2)$ colors such that each node has at most $\alpha \cdot \deg(v)$ neighbors of the same color. If the graph is initially equipped with a proper $q$-coloring, the algorithm requires $O(\log^* q)$ communication rounds.
\end{lemma}

Usually, when computing a $(\Delta+1)$-coloring, this is done in several phases. One starts with all nodes being uncolored and in each phase, a given partial proper coloring of the nodes is extended by coloring some of the previously uncolored nodes. When solving $(\Delta+1)$-coloring in this way, the problem that needs to be solved in each step is a list coloring problem or more precisely a $(\mathit{degree}+1)$-list coloring problem: Each node $v$ is initially given a list $L_v$ of available colors and in the end, the nodes have to be colored properly such that each node $v$ is colored with a color from its list $L_v$. In a $(\mathit{degree}+1)$-list coloring instance, the list $L_v$ of each node $v$ is of size $|L_v|\geq\deg(v)+1$. The notion of defective coloring has been extended to list colorings implicitly in \cite{Kuhn20} and explicitly in \cite{FK23}. Formally, in \emph{list defective coloring}, each node $v$ obtains a color list $L_v$ together with a defect function $d_v:L_v \to \mathbb{N}_0$ that assigns non-negative integer values to the colors in $L_v$. The coloring has to assign a color $x_v\in L_v$ to each node $v$ such that the number of neighboring nodes $u$ that get assigned color $x_u=x_v$ is at most $d_v(x_v)$. While standard defective colorings can be used to recursively divide non-list coloring instances, list defective colorings can be used to recursively divide list coloring instances (for details, see the discussion on color space reduction below).

Finally, we need to define the notion of \emph{arbdefective colorings}, which was introduced in \cite{BarenboimE10}. In \cite{BarenboimE10}, a $d$-arbdefective $c$-coloring of a graph was defined as a coloring of the nodes with $c$ colors such that the graph induced by each color class has \emph{arboricity} at most $d$. We use a slightly different definition, which is locally checkable and which has turned out to be algorithmically more convenient (see, e.g., \cite{BarenboimE11,BarenboimEG18,Kuhn20,hideandseek}). We define a $d$-arbdefective $c$-coloring of a graph as a $c$-coloring of the nodes \emph{together} with an orientation of the monochromatic edges so that every node has at most $d$ outneighbors of the same color. We further note that this only guarantees that each color class has pseudoarboricity at most $d$ (and thus arboricity at most $d+1$). One could strengthen the definition to imply arboricity at most $d$ by requiring the edge orientation of each color class to be acyclic. We do not need to make this additional restriction.

List arbdefective colorings are defined in the natural way: Each node $v$ must be colored with a color $x_v$ from its list $L_v$, the monochromatic edges are oriented such that $v$ has at most $d_v(x_v)$ \textit{outneighbors} of color $x_v$. Unlike standard defective colorings, $d$-arbdefective colorings with $\lceil\frac{\Delta+1}{d+1}\rceil$ colors can be computed by a simple sequential greedy algorithm and as a result, efficient distributed arbdefective coloring algorithms with a small number of colors are known. In particular, in \cite{BarenboimEG18}, it is shown that a $d$-arbdefective coloring with $O(\Delta/d)$ colors can be computed deterministically in $O(\Delta/d+\log^* n)$ rounds. The main reason why for graphs of bounded neighborhood independence, we have faster deterministic coloring algorithms than for general graphs, is the following observation, which was first made in \cite{BarenboimE11}. If $G$ is a graph with neighborhood independence $\theta$, then any $d$-arbdefective coloring of $G$ at the same time is a $(2d+1)\cdot\theta$-defective coloring of $G$. For graphs of bounded neighborhood, we therefore have efficient algorithms to compute $d$-defective colorings with $O(\Delta/d)$ colors. All the fast $f(\Delta)+O(\log^* n)$-round coloring algorithms that specifically work for families of graphs with bounded neighborhood independence are in particular based on this fact~\cite{BarenboimE11,Kuhn20,BalliuKO20,BalliuBKO22}.

We conclude this short discussion of defective colorings with the following notion of \emph{slack}, which specifies how loose a given list defective or list arbdefective coloring instance is.

\begin{definition}[Slack]\label{def:slack}
  Consider a graph $G=(V,E)$ and a list defective or list arbdefective coloring instance with lists $L_v$ and defect functions $d_v$ for all $v\in V$. We say that the instance has slack $S$ if
  \[
    \sum_{x \in L_v} (d_v(x) + 1) > S \cdot \deg(v).
  \]
\end{definition}

\subsection{Recursive Subproblems}\label{sec:recsubproblems}

\begin{figure}[t]
  \centering
  \tikzset{every picture/.style={line width=0.75pt}}
\definecolor{reddish}{rgb}{0.65,0,0}

\begin{tikzpicture}[x=0.75pt,y=0.75pt,yscale=-1,xscale=1,scale=0.82, every node/.style={scale=0.82}]

\draw [<-]   (100,99) -- (220.5,99) ;
\draw [<-]   (270,120) -- (270,200) ;
\draw [<-]   (318,220) -- (470, 220) ;
\draw [<-]   (318,99) -- (470, 99) ;

\draw (20,85) -- (95,85) -- (95,111) -- (20,111) -- cycle  ;
\draw (55,98) node [anchor=center][inner sep=0.75pt]   [align=center] {$P_{A}$(2, $C$)};
\draw (220,85) -- (317,85) -- (317,112) -- (220,112) -- cycle  ;
\draw (265,99) node [anchor=center][inner sep=0.75pt]   [align=center] {$P_{A}(64 \cdot \theta, C)$};
\draw (162,85) node [anchor=center][inner sep=0.75pt]   [align=center] {$O\left( \theta^{2} \right)$ instances};
\draw    (220,205) -- (317,205) -- (317,235) -- (220,235) -- cycle  ;
\draw (265,220) node [anchor=center][inner sep=0.75pt]   [align=center] {$P_{D}\left(32 \cdot \theta ,C^{\varepsilon} \right)$};
\draw (255,165) node [anchor=center][inner sep=0.75pt]  [rotate=90] [align=center] {$1$ instance};
\draw    (475,85) -- (585,85) -- (585,112) -- (475,112) -- cycle  ;
\draw (530,99) node [anchor=center][inner sep=0.75pt]   [align=center] {$P_{A}\left( 2,\ C^{1-\varepsilon} \right)$};
\draw (395,85) node [anchor=center][inner sep=0.75pt]  [rotate=0] [align=center] {$C^{\varepsilon}$ parallel instances};
\draw (395, 205) node [anchor=center][inner sep=0.75pt]  [rotate=0] [align=center] {$O(\log^3 \Delta )$ instances };
\draw (160,110) node [anchor=center][inner sep=0.75pt]  [color=reddish  ,opacity=1 ] [align=center] {\Cref*{lemma:SlackGeneration}};
\draw (285,165) node [anchor=center][inner sep=0.75pt]  [color=reddish  ,opacity=1 ,rotate=90] [align=center] {\Cref*{lemm:colorSpaceReduction}};
\draw (395,110) node [anchor=center][inner sep=0.75pt]  [color=reddish  ,opacity=1, rotate=0 ] [align=center] {\Cref*{lemm:colorSpaceReduction}};
\draw (395,235) node [anchor=center][inner sep=0.75pt]  [color=reddish  ,opacity=1 ,rotate=0] [align=center] {\Cref*{thm:ArbToDef}};
\draw    (475,205) -- (585,205) -- (585,235) -- (475,235) -- cycle  ;
\draw (530, 220) node [anchor=center][inner sep=0.75pt]   [align=center] {$P_{A}\left( 2,\ C^{\varepsilon} \right)$};

\end{tikzpicture}
  \caption{The recursive structure for solving list arbdefective coloring instances $P_A(2,C)$ with slack $\geq 2$ and colors from a space of size $C$. This observation leads to the statement $T_A(2, C) \leq O(\theta^2) \cdot (\log^3 \Delta \cdot T_A(2, C^\varepsilon) + T_A(2, C^{1-\varepsilon}))$  (cf. \Cref{eq:splitting}). Choosing $\varepsilon$ to balance the recursion (cf.\ \Cref{sec:technicalProofs}) or setting $\varepsilon = 1/2$ (depending on the value of $\theta$) results in the complexity claimed in \Cref{thm:mainNeighborhood}.}
  \label{fig:boundedNeighScheme}
\end{figure}

First note that list arbdefective coloring generalizes $(\Delta+1)$-list coloring in the following way. First note that $(\Delta+1)$-coloring is clearly a special case of the $(\mathit{degree}+1)$-list coloring problem. Further, if every node $v$ is given a list $L_v$ of colors and the allowed defect or arbdefect for every $v\in V$ and every $x\in L_v$ is set to $d_v(x)=0$, then the given list defective or list arbdefective coloring instance is a standard list coloring instance. The instance corresponds to a $(\mathit{degree}+1)$-list coloring instance if and only if its slack is at least $1$. If we have an algorithm to solve list arbdefective coloring instances with slack $1$, we can therefore use this algorithm to solve $(\mathit{degree}+1)$-list coloring instances and thus also to solve $(\Delta+1)$-coloring.

Our algorithm to achieve the upper bound of \Cref{thm:mainNeighborhood} is a recursive algorithm to solve list arbdefective coloring instances. As part of the recursive structure, we also solve some list defective coloring instances. Both problems, we parameterize by two quantities: the slack $S$ (as given by \Cref{def:slack}) and the size of the color space $C$. The size of the color space is equal to $C$ if all the color lists $L_v$ of the nodes are composed of colors from a set of size $C$ (e.g., if the lists only use colors $1,\dots,C$). For the remainder of the analysis, we use the following notation.

We use $P_A(S, C)$ to denote the family of list arbdefective coloring instances with slack at least $S$ and with a color space of size at most $C$. Similarly, we use $P_D(S,C)$ to denote the corresponding list defective coloring problems. Recall that we defined (list) arbdefective coloring to be a coloring of the nodes together with an orientation of the monochromatic edges. Solving an instance of $P_A(S,C)$ therefore means to compute the coloring and the required edge orientations. The optimal round complexities of solving problems of the form $P_A(S, C)$ and $P_D(S, C)$ deterministically in the \LOCAL model are denoted by $T_A(S, C)$ and $T_D(S, C)$, respectively. Unless stated otherwise, we assume throughout the recursive algorithm that the graph is equipped with an initial proper $q$-coloring (where initially a proper $q=O(\Delta^2)$-coloring can be computed in $O(\log^* n)$ rounds, and all recursive subproblems inherit this coloring with $q=O(\Delta^2)$); the $O(\log^* q)$ round overhead is made explicit in the stated bounds. Before continuing with the discussion of how we solve $P_A(S,C)$ and $P_D(S,C)$, we make a remark regarding the parameterization by the slack $S$ and the color space size $C$. Typically, the parameterization of both the number of allowed colors and the time complexity of local coloring problems is mostly done as a function of the maximum degree $\Delta$ of the graph. Naturally, the time complexity as a function of $\Delta$ decreases if the number of available colors increases as a function of $\Delta$. Consider the problem of coloring with $f(\Delta)\geq\Delta+1$ colors. In this case (if we do not consider list coloring), the size of the color space is $C=f(\Delta)$ and the slack is $S=f(\Delta)/(\Delta+1)$. For any of the standard coloring problems, parameterization by $S$ and $C$ thus directly implies a parameterization of both the number of colors and time complexity as a function of $\Delta$.

Note that in general $P_A(S,C)$ is an easier problem than $P_D(S,C)$, i.e., we have $T_A(S,C)\leq T_D(S,C)$ for any $S$ and $C$. Any solution to $P_D(S,C)$ automatically also is a solution to $P_A(S,C)$, even with an arbitrary orientation of the edges.\footnote{In some cases, it is also known that $P_D(S,C)$ is strictly harder than $P_A(S,C)$~\cite{BalliuHLOS19}.} Further, as already remarked above, if we have $d_v(x)=0$ for all $v\in V$ and all $x\in L_v$, then both $P_D(S,C)$ and $P_A(S,C)$ are equivalent to the standard proper list coloring problem with lists $L_v$ of size $|L_v|>S\cdot\deg(v)$. The highlevel recursive structure of our algorithm is depicted in \Cref{fig:boundedNeighScheme}. In the following, we briefly describe each of the steps.

\subsection{Slack Generation}\label{sec:intro:slackgeneration}

List coloring problems naturally become easier if the color lists and thus the slack of the problem gets larger. A simple way to generate slack is the following. Starting from a problem with small slack, use some defective coloring to divide the graph into several parts of significantly smaller degree. Now, when only coloring one of the parts, but using the full list of colors, the slack of the instance to solve is much larger. If we have a list coloring algorithm, we can sequentially go through the different parts and color them. This method was first introduced in \cite{barenboim16sublinear} and it has subsequently been optimized and used in \cite{fraigniaud16local,Kuhn20,BalliuKO20,FK23}. This technique can not only be applied to standard list coloring, but also to the more general list arbdefective coloring problem $P_A(S,C)$, i.e, one can solve a list arbdefective coloring instance $P_A(S,C)$ by iteratively solving several instances of $P_A(S',C)$ for $S'>S$. If we use the defective coloring algorithm of \cite{Kuhn2009} and if the initial slack is $S\geq S_0$ for some constant $S_0>1$, the number of required slack $S'$-instances is $O(S'^2)$. If the initial slack is $S=1$, the method only allows to reduce the maximum uncolored degree by some constant factor and one has to repeat the whole process $O(\log\Delta)$ times to fully color the graph. The number of slack $S'$-instances therefore becomes $O(S'^2\cdot\log\Delta)$ in this case. This slack generation technique is formally stated and proven in the following lemma. In our algorithm, we use the technique to first reduce $P_A(1,C)$ to $O(\log\Delta)$ instances of $P_A(2,C)$ and we then solve $P_A(2,C)$ recursively.

\begin{restatable}{lemma}{restateLemmaSlackGeneration}\label{lemma:SlackGeneration}
    Let $G=(V,E)$ be a graph that is equipped with a proper $q$-vertex coloring. For any $S \geq 1$, we have
    \begin{eqnarray}
         T_A(1,C) & \leq & O(\log\Delta)\cdot \big(T_A(2, C)+\log^*\Delta\big) + O(\log^* q)\label{eq:slackGeneration1}\\
          T_A(2, C) & \leq & O(S^2) \cdot T_A\left(S, C \right) + O(\log^* q). \label{eq:slackGeneration2}
        \end{eqnarray}
\end{restatable}
\begin{proof}
  The proof of both bounds \eqref{eq:slackGeneration1} and \eqref{eq:slackGeneration2} are based on the same idea, but they differ in some of the details. In the following, whenever we refer to only one of the cases, we refer to Case \eqref{eq:slackGeneration1} or Case \eqref{eq:slackGeneration2}. In Case \eqref{eq:slackGeneration1}, we assume that we are given a list arbdefective coloring instance on $G$ of the form $P_A(1,C)$ and in Case \eqref{eq:slackGeneration2}, we assume that we are given a list arbdefective coloring instance on $G$ of the form $P_A(2,C)$. In both cases, we assume that each node $v\in V$ has a color list $L_v$ with allowed arbdefects $d_v(x)$ for $x\in L_v$. In Case \eqref{eq:slackGeneration1}, we have $\sum_{x\in L_v} (d_v(x)+1)>\deg(v)$ and in Case \eqref{eq:slackGeneration2}, we have $\sum_{x\in L_v} (d_v(x)+1)>2\cdot\deg(v)$. In Case \eqref{eq:slackGeneration1}, we assume that we have access to an algorithm to solve list arbdefective coloring problems of the form $P_A(2,C)$ and in Case \eqref{eq:slackGeneration2}, we assume that we have access to an algorithm to solve list arbdefective coloring problems of the form $P_A(S,C)$.

  The highlevel idea is to generate some slack by using a defective coloring to partition $G$ into several subgraphs of smaller degree. In Case \eqref{eq:slackGeneration1}, we set $\eps:=1/4$ and in Case \eqref{eq:slackGeneration2}, we set $\eps:=1/S$. In Case \eqref{eq:slackGeneration1}, we first use a standard algorithm of \cite{Linial1987} to compute a proper $O(\Delta^2)$-coloring of $G$ in $O(\log^* q)$ rounds. We then use \Cref{lemm:defColorBlackBox} to compute a defective coloring of $G$ with $O(1/\eps^2)$ colors such that for every node $v\in V$, at most $\eps\cdot\deg(v)$ neighbors are colored with the same color. Let $\set{1,\dots,Q}$ be the set of those $O(1/\eps^2)$ colors. For each $i\in \set{1,\dots,Q}$, we define $V_i$ to be the nodes with color $i$ and $G_i:=G[V_i]$ to be the subgraph induced by $V_i$. The algorithm now consists of $Q$ phases. Initially, all nodes are uncolored (w.r.t.\ the list arbdefective instance $P_A(2,C)$ or $P_A(S,C)$ that we have to solve). Then, in every phase, some nodes are colored and the monochromatic edges become oriented. In Case \eqref{eq:slackGeneration1}, the \emph{active} nodes in phase $i$ are all nodes in $V_i$ for which at the beginning of phase $i$, at least $\deg(v)/2$ neighbors are still uncolored. In Case \eqref{eq:slackGeneration2}, the \emph{active} nodes in phase $i$ are all nodes in $V_i$. Let $A_i\subseteq V_i$ be the set of active nodes in phase $i$. In phase $i$, we then solve a list arbdefective coloring problem on the subgraph $G[A_i]$ induced by the nodes in $A_i$. For each $v\in A_i$ and every $x\in L_v$, let $\alpha_v(x)$ be the number of neighbors of $v$ that have already been colored with color $x$ prior to phase $i$. In the list arbdefective coloring problem of phase $i$, each node $v\in A_i$ uses a color list $L_v'\subseteq L_v$ that consists of all colors $x\in L_v$ for which $\alpha_v(x)\leq d_v(x)$ and $v$ uses defect $d'_v(x):=d_v(x)-\alpha_v(x)$ for color $x$. We need to show that this list arbdefective coloring problem has the desired slack. Note that for every discarded color $x\in L_v\setminus L'_v$, we have $\alpha_v(x)\geq d_v(x)+1$ and thus $d_v(x)+1-\alpha_v(x)\leq 0$. In Case \eqref{eq:slackGeneration1}, we have
  \[
    \sum_{x\in L_v'}(d_v'(x)+1) \geq \sum_{x\in L_v}\big(d_v(x)+1 -\alpha_v(x)\big) \geq
    \sum_{x\in L_v} (d_v(x)+1) - \frac{\deg(v)}{2} > \deg(v) - \frac{\deg(v)}{2} = \frac{\deg(v)}{2}.
  \]
  In the second inequality, we use that only nodes $v$ that still have at least $\deg(v)/2$ uncolored neighbors are active, which implies that $\sum_{x\in L_v}\alpha_v(x)\leq \deg(v)/2$. Because the degree of $v$ in $G[A_i]$ is at most $\eps\cdot\deg(v)=\deg(v)/4$, the list arbdefective coloring instance of phase $i$ has slack at least $2$. We can therefore use our list arbdefective coloring algorithm $P_A(2,C)$ that we assume in this case to solve the problem of phase $i$. In Case \eqref{eq:slackGeneration2}, we have
  \[
    \sum_{x\in L_v'} (d_v'(x)+1) \geq \sum_{x\in L_v}\big(d_v(x)+1 -\alpha_v(x)\big) \geq
    \sum_{x\in L_v} (d_v(x)+1) - \deg(v) > 2\deg(v) - \deg(v) = \deg(v).
  \]
  Because the degree of $v$ in $G[A_i]$ is at most $\eps\cdot\deg(v)=\deg(v)/S$, the list arbdefective coloring instance of phase $i$ has slack at least $S$. Here, we can therefore use our given list arbdefective coloring algorithm $P_A(S,C)$ to solve the instance of phase $i$. In the end, all monochromatic edges between nodes that are colored in different phases are oriented from the node colored later to the node colored earlier. The monochromatic edges between nodes that are colored in the same phase are oriented by applying the existing list arbdefective coloring algorithm $P_A(2,C)$ or $P_A(S,C)$, respectively. By construction, for each node $v$ that in the end is colored with some color $x\in L_v$, the set of outneighbors with color $x$ consists of the $\alpha_v(x)$ neighbors that are colored with color $x$ in phases prior to the phase in which $v$ is colored and of the at most $d'_v(x)=d_v(x)-\alpha_v(x)$ outneighbors with color $x$ that $v$ gets in phase $i$. At the end, if some node $v$ gets colored with color $x\in L_v$, the above algorithm therefore guarantees that $v$ has at most $d_v(x)$ outneighbors of color $x$. In Case \eqref{eq:slackGeneration1}, the algorithm only colors a subset of all nodes, in Case \eqref{eq:slackGeneration2}, the algorithm colors all the nodes. In Case \eqref{eq:slackGeneration1}, for each node that remains uncolored, at least half of the neighbors are colored.

  The idea in Case \eqref{eq:slackGeneration1} therefore is to recurse so that after $O(\log\Delta)$ iterations, all the nodes are colored. Consider some node $v$ that remains uncolored and let $\beta_v$ be the number of neighbors of $v$ that are already colored and let $\beta_v(x)$ be the number of neighbors colored with a specific color $x\in L_v$. Let $L_v'\subseteq L_v$ be the list consisting of all colors $x\in L_v$ for which $\beta_v(x) \leq d_v(x)$ and define $d'_v(x):=d_v(x)-\beta_v(x)$ for all $x\in L_v'$. We use lists $L_v'$ and defects $d'_v(x)$ for the list arbdefective coloring problem on the remaining uncolored nodes. For every discarded color $x\in L_v\setminus L'_v$, we again have $d_v(x)+1-\beta_v(x)\leq 0$. This problem still has slack $1$ because
  \[
    \sum_{x\in L_v'} (d_v'(x)+1) \geq \sum_{x\in L_v} \big(d_v(x)+1 - \beta_v(x)\big) =
    \sum_{x\in L_v} (d_v(x)+1) - \beta_v > \deg(v) - \beta_v.
  \]
  The original list arbdefective coloring problem can be solved by solving this problem on the uncolored nodes and by orienting all monochromatic edges from later to earlier colored nodes.

    In Case \eqref{eq:slackGeneration1}, we can therefore solve the given $P_A(1,C)$ instance by solving $O(1/\eps^2)\cdot\log\Delta = O(\log\Delta)$ instances of $P_A(2,C)$ and in Case \eqref{eq:slackGeneration2}, we can solve the given $P_A(2,C)$ instance by solving $O(1/\eps^2)=O(S^2)$ instances of $P_A(S,C)$. In addition, we have to compute the defective coloring that we use to partition $G$ into parts of smaller degree. In Case \eqref{eq:slackGeneration1}, we first use $O(\log^* q)$ rounds to compute a proper $O(\Delta^2)$ coloring. We then spend $O(\log^*\Delta)$ rounds in each of the $O(\log\Delta)$ iterations to compute the defective coloring (by using \Cref{lemm:defColorBlackBox}). In Case \eqref{eq:slackGeneration2}, we can apply \Cref{lemm:defColorBlackBox} directly, which requires $O(\log^* q)$ rounds.
\end{proof}

\subsection{Color Space Reduction}\label{sec:intro:colspacereduction}
List defective coloring was originally introduced in \cite{Kuhn20,FK23} as a natural problem to recursively reduce the color space of a list (arbdefective) coloring problem. In the second step, we use exactly this idea to a given instance of list arbdefective coloring $P_A(S, C)$. We partition the color space of size $C$ into $C^\varepsilon$ parts of size $C^{1-\varepsilon}$ for some suitably chosen $0 < \varepsilon < 1$ (in this overview, we for simplicity assume that $C^\varepsilon$ and $C^{1-\varepsilon}$ are integers). The partition of the color space is done in an arbitrary way. We then assign each node $v$ to one of the $C^\varepsilon$ color subspaces. When doing this, the color list $L_v$ of node $v$ (of the list arbdefective coloring problem that we have to solve) is restricted to only the colors in the chosen color space. For the assignment of nodes to color subspaces, we set up a list defective coloring problem instance $P_D(\Theta(\theta), C^\varepsilon)$. In this list defective coloring instance, the allowed defect for a node $v$ when choosing some color subspace is proportional to the sum of the arbdefects $d_v(x)$ of $v$ (in the given list arbdefective coloring problem $P_A(S,C)$ that we need to solve) corresponding to colors $x\in L_v$ that are in the given color subspace. This results in $C^{\varepsilon}$ instances of $P_A(2, C^{1-\varepsilon})$ that operate on disjoint color spaces and can therefore be solved in parallel. The slack $S$ of the original list arbdefective coloring problem has to be chosen sufficiently large to compensate for the loss in slack in this partitioning step (which is equal to the slack of the defective coloring instance that we use to divide the color space). This color space reduction step is formally done in \Cref{lemm:colorSpaceReduction}. The choice of $\varepsilon$ depends on the value of the neighborhood independence $\theta$; for our main result, we either choose $\varepsilon = 1/2$ (if $\theta$ is large) or balance the recursive costs by choosing $\eps$ as analyzed in \Cref{sec:technicalProofs} (if $\theta$ is small). If $\theta$ is small, the time that we need to compute the list defective coloring to divide the color space is more significant than the loss in slack in this step. At the cost of having a somewhat larger recursion depth, we can make this step more efficient by choosing $\varepsilon$ and thus $C^{\varepsilon}$ small. The list defective coloring instance to assign the color subspaces then itself has a smaller color space over which it operates. The details are given by the following lemma, proven below. To be notation-wise consistent, we partition the color space $\calC$ into (almost) equally sized subspaces $\calC_1, \calC_2, \ldots, \calC_p$ and assign each node to exactly one of these subspaces (by invoking a list defective coloring instance).

\begin{restatable}{lemma}{restateLemmaColorSpaceReduction}\label{lemm:colorSpaceReduction}
  Let $1 \leq \sigma \leq S$ and $p \in \{1, \ldots, C\}$, then we have
    \begin{align}
        T_A(S, C) \leq T_D\left(\sigma, p \right) + T_A\left(\frac{S}{\sigma}, \left\lceil \frac{C}{p} \right\rceil \right).
    \end{align}
\end{restatable}
\begin{proof}
  Let $\calC$ be the color space of the given instance and partition it arbitrarily into $p$ disjoint subspaces $\calC_1,\dots,\calC_p$, each of size at most $\lceil C/p\rceil$.
  We start by defining the list defective coloring problem that we use to assign the $p$ subspaces $\calC_i$ to the nodes. The color space of this 'new' list defective coloring problem is $\set{1,\dots,p}$. The color list $X_v\subseteq\set{1,\dots,p}$ of node $v$ contains all $i\in \set{1,\dots,p}$ for which $L_{v,i}:=L_v\cap \calC_i\neq\emptyset$, i.e., all color subspaces for which $v$ has at least one color it can pick. The defect $d_{v,i}$ of node $v$ for color (subspace) $i\in X_v$ is defined as
  \[
    d_{v,i} := \left\lfloor \sigma \cdot \deg(v) \cdot \frac{\sum_{x \in L_{v, i}}(d_v(x) + 1)}{\sum_{x \in L_{v}}(d_v(x) + 1)} \right\rfloor.
  \]
  We first show that the list $X_v$ and defects $d_{v,i}$ define a list defective coloring problem with slack $\sigma$, i.e., a list defective coloring problem of the form $P_D(\sigma, p)$. Using the elementary inequality $\lfloor z\rfloor+1>z$ for every real $z$, for every node $v\in V$, we have
    \begin{equation*}
        \sum_{i\in X_v} (d_{v, i} + 1) > \sigma \cdot \deg(v) \cdot \frac{\sum_{i\in X_v} \sum_{x \in L_v \cap \calC_i}(d_v(x) + 1)}{\sum_{x \in L_{v}}(d_v(x) + 1)} =  \sigma \cdot \deg(v).
    \end{equation*}

    For the following, assume that we have an assignment of color subspaces such that if a node $v$ is assigned subspace $\calC_i$ for $i\in X_v$, then $v$ has at most $d_{v,i}$ neighbors that are also assigned color subspace $\calC_i$. For each $i\in \set{1,\dots,p}$, let $V_i$ be the set of nodes assigned to subspace $\calC_i$. We proceed by showing that the arbdefective coloring instance on $G[V_i]$ is of the form $P_A(S/\sigma, \lceil C/p \rceil)$. For a node $v\in V_i$, its degree in $G[V_i]$ is at most $d_{v,i}$. We therefore have to show that $\sum_{x\in L_{v,i}}(d_v(x)+1)>\frac{S}{\sigma}\cdot d_{v,i} \geq \frac{S}{\sigma}\cdot \deg_{G[V_i]}(v)$.
    If $d_{v,i}=0$ (which in particular holds if $\deg(v)=0$), then because $i\in X_v$, we have $L_{v,i}\neq\emptyset$. Since $d_v(x)\ge 0$ for all $x\in L_v$, this implies
    \[
      \sum_{x\in L_{v,i}}(d_v(x)+1) \geq 1 > 0 = \frac{S}{\sigma}\cdot d_{v,i}.
    \]
    If $d_{v,i}>0$ (and thus $\deg(v)>0$), we have from the definition of $d_{v,i}$ that
    \begin{align*}
        \sum_{x \in L_{v, i}} (d_v(x) + 1) \geq \frac{d_{v, i} \cdot \sum_{x \in L_v} (d_v(x) + 1)}{\sigma \cdot \deg(v)} > \frac{d_{v,i}\cdot (S\cdot \deg(v))}{\sigma \cdot \deg(v)} = \frac{S}{\sigma} \cdot  d_{v, i},
    \end{align*}
    where we used the definition of slack $S$, i.e., that $\sum_{x \in L_v} (d_v(x) + 1) > S \cdot \deg(v)$, and multiplied by $\frac{d_{v,i}}{\sigma\cdot\deg(v)}>0$. In both cases, the instance on $G[V_i]$ has slack at least $S/\sigma$. Since the $p$ subproblems use disjoint palettes $\calC_1, \dots, \calC_p$ of size at most $\lceil C/p\rceil$, they can be solved in parallel in time $T_A(S/\sigma, \lceil C/p\rceil)$. Thus, solving $P_A(S, C)$ takes time $T_D\left(\sigma, p \right) + T_A\left(\frac{S}{\sigma}, \left\lceil \frac{C}{p} \right\rceil \right)$.
\end{proof}

\subsection{List Defective Coloring With Small Neighborhood Independence}\label{sec:intro:listdefective}

The last step is the only one that requires small neighborhood independence and it is the main technical contribution towards proving \Cref{thm:mainNeighborhood}. As mentioned in \Cref{sec:defective}, in graphs of neighborhood independence $\theta$, any $d$-arbdefective coloring is also a $(2d+1)\cdot\theta$-defective coloring. The reason is as follows. Assume that in a given $d$-arbdefective coloring, some node $v$ is assigned color $x$ and let $A$ be the set of neighbors of $v$ that are also colored with color $x$. Because in a $d$-arbdefective coloring, all monochromatic edges are oriented and any node has at most $d$ outneighbors of the same color, the subgraph induced by $A$ is an oriented graph with outdegree at most $d$. This subgraph can therefore be colored with at most $2d+1$ colors. Because the neighborhood independence of our graph is at most $\theta$ and because all nodes in $A$ are neighbors of a single node (of $v$), any independent set of $G[A]$ has size at most $\theta$. Hence, we have $|A|\leq (2d+1)\cdot\theta$. By using the arbdefective coloring algorithm of \cite{BarenboimEG18}, a $(d/(2 \theta)-1/2)$-arbdefective coloring with $O(\Delta\theta/d)$ colors and thus a $d$-defective coloring with $O(\Delta\theta/d)$ colors can be computed in $O(\Delta\theta/d + \log^* n)$ rounds.

In this paper, we extend this idea to the list versions of the problems. That is, in graphs of small neighborhood independence, our goal is to use an algorithm for computing list arbdefective colorings to solve some given list defective coloring instance (with larger slack). The main challenge to achieve this is that in list defective colorings, different nodes can have different defects for the same colors. This means that in the above argument, the nodes in the set $A$ could use much larger defects for color $x$ and they could therefore have much larger outdegrees than what $v$ uses for color $x$. As a result, we unfortunately cannot set up a single list arbdefective coloring instance to solve a given list defective coloring instance. The problem is relatively easy to overcome if for each node $v$, $v$ uses approximately the same defect for all the colors in its color list $L_v$. One can then iterate over $O(\log\Delta)$ defect classes (starting with large defects) and only color the nodes that use approximately the same defects in a single list arbdefective coloring instance. Because when $v$ gets its color, the neighbors that use larger defects have already been colored, $v$ can control the number of such neighbors that have the same color. One can enforce that the defects of all colors in the list of a node $v$ are equal by sacrificing an $O(\log\Delta)$-factor in the slack of the problem (each node only keeps the colors for one of the defect classes).

In order to avoid the additional $\Theta(\log\Delta)$ factor in the slack of the list defective coloring, we extend this approach by combining it with the general idea of slack generation as described in \Cref{sec:intro:slackgeneration}. We can temporarily increase the slack by an $O(\log\Delta)$ factor by using the defective coloring algorithm of \cite{Kuhn2009} to divide the graph into $\Theta(\log^2\Delta)$ subgraphs in which the degree of each node is by an $\Omega(\log\Delta)$ factor smaller. Note however that if we do not have the additional $O(\log\Delta)$ factor in the slack, but if we nevertheless just keep one defect class in each list, the slack of the actual instance to solve still becomes smaller than $1$. We therefore still have to use the full list. We can however show that by iterating over the $O(\log\Delta)$ defect classes and over the subgraphs of the smaller degrees for each defect class, for each node, the slack becomes large enough in one of the $O(\log\Delta)$ iterations where we consider the node. This allows to solve a single list defective coloring instance $P_D(O(\theta),C)$ by iteratively solving $O(\log^3\Delta)$ list arbdefective coloring instances $P_A(2, C)$. This is formally stated in the following theorem, whose proof appears in \Cref{sec:arbtodef}.

\begin{restatable}{theorem}{restateforth}\label{thm:ArbToDef}
  Let $G=(V,E)$ be an $n$-node graph of maximum degree $\Delta$ that is properly $q$-colored. Then, for $S \geq 1$, we have
  \begin{align} \label{eq:defToArbdef}
      T_D(16 \cdot \theta \cdot S, C) \leq O(\log^3 \Delta) \cdot T_A(S, C) + O(\log^* q).
  \end{align}
\end{restatable}
\newpage

\subsection{Defective Hypergraph 2-Edge Coloring Lower Bound}\label{sec:intro:lower}

The proof of \Cref{thm:lowerbound} builds on a lower bound argument of \cite{BalliuHLOS19}. There, it is shown that computing a $(d-2)$-defective $2$-coloring in $d$-regular trees requires $\Omega(\log n)$ deterministic rounds. The core idea is a reduction from the sinkless orientation problem in $2$-colored regular bipartite graphs, which itself has a deterministic $\Omega(\log n)$ lower bound~\cite{ChangHLPU18}. Our contribution is to extend this reduction to defective $2$-edge colorings of hypergraphs of small rank $r$ (or equivalently, to $2$-defective vertex colorings of line graphs of rank-$r$ hypergraphs).

The argument in \cite{BalliuHLOS19} relies on the complexity classification of LCL problems in bounded-degree graphs \cite{chang16exponential}: every such problem either requires $\Omega(\log n)$ deterministic rounds or can be solved deterministically in time $O(\log^* n)$. In the latter case, one first computes a proper distance-$k$ coloring with $O(1)$ colors for some (problem-dependent) constant $k$. Afterwards the problem can be solved in $O(1)$ rounds by just using those colors and without using unique IDs. To rule out the $O(\log^* n)$ case for $(d-2)$-defective $2$-coloring in $d$-regular trees, one can thus assume the existence of a constant-time algorithm that works given such a distance-$k$ coloring and lead this to a contradiction.

The reduction proceeds as follows. For a given $\Delta$-regular $2$-colored graph $G$, where $\Delta=d(d-1)^{2k-1}$, we first replace each edge with two half-edges. Each node with its $\Delta$ incident half-edges is then replaced by a $d$-regular tree of height $2k$ whose leaves correspond to the half-edges. When two nodes $u,v$ are connected in $G$, we merge the corresponding leaves in their trees and add some additional edges to restore degree $\Delta$. This yields a graph where all nodes have degree $d$ or $1$, and on which we can run the assumed $(d-2)$-defective $2$-coloring algorithm (say with colors black and white). By controlling the distance-$k$ coloring, one can enforce that all nodes on one side of the bipartition of $G$ output black and the nodes on the other side output white. If a $d$-regular tree is complete up to some even distance $2\ell$ from the root, one can show that in any valid $(d-2)$-defective $2$-coloring, one of the nodes at distance $2\ell$ must have the same color as the root. Because the leaf nodes of the $d$-regular trees of height $2k$ from both sides are merged and the roots on the two sides use opposite colors, this implies that each root node sees both colors among the nodes at distance $2k$. On the original $\Delta$-regular bipartite graph, this induces a $2$-coloring of the edges such that each node is incident to edges of both colors and in $2$-colored bipartite graphs, such an edge coloring is equivalent to a sinkless orientation, completing the reduction.

To adapt this to hypergraphs of rank $r$, we use the standard representation of a $d$-regular rank-$r$ hypergraph as a biregular bipartite graph: degree-$r$ nodes correspond to hyperedges, and degree-$d$ nodes to original vertices. A $\delta$-defective edge-coloring of the hypergraph then corresponds to coloring the degree-$r$ nodes such that for each such node, at most $\delta$ of the degree-$r$ nodes at distance two share the same color. In the reduction, each node with $\Delta$ half-edges is replaced by a $(r,d)$-biregular tree of height $4k$, rooted at a degree-$r$ node. In such trees, every degree-$r$ node has $r(d-1)$ degree-$r$ nodes at distance two, of which at least $d$ must have the opposite color if the defect is bounded by $r(d-1)-d$. Because in a rooted $(r,d)$-biregular tree, each degree-$r$ node has only $d-1$ nodes at distance $2$ that are either closer to or equally far from the root, at least one node with the opposite color must be one of the children's children. This guarantees that at every fourth level, some node shares the root’s color. By ensuring opposite root colors for the two sides of the bipartition, the same argument as before induces an edge coloring of the bipartite graph with $2$ colors that forces both colors at every node—again equivalent to a sinkless orientation. The full proof is given in \Cref{sec:lowerbound}.
\newpage

\section{List Defective Coloring Algorithm}\label{sec:arbtodef}

As discussed in \Cref{sec:intro:listdefective}, any $d$-arbdefective coloring of a graph $G$ of neighborhood independence $\theta$ is also a defective coloring of $G$ with defect at most $(2d+1)\cdot\theta$. The goal of this section is to generalize this relation to the list versions of the problems and to show how one can solve a given list defective coloring instance with slack at least $c\cdot \theta$ for a sufficiently large constant $c$ by using a given list arbdefective algorithm (and thus to prove \Cref{thm:ArbToDef}). As stated in \Cref{sec:intro:listdefective}, the relation between list defective and list arbdefective coloring is not as straightforward as in the non-list setting, primarily because the allowed defect for each color might be very different for different nodes.

To overcome this problem, we partition each list into colors of similar defect and color the nodes in different phases, where in each phase, we use colors with approximately equal defects. To ensure that the used lists are still sufficiently large, we use the idea of \emph{slack generation} (similar to \Cref{sec:intro:slackgeneration}) in each phase. In the following, we describe our algorithm in detail.

\subparagraph{Algorithm Description.}
The aim of the algorithm is to solve a list defective coloring instance of the form $P_D(4 \cdot (\theta+1) \cdot (S + 1), C)$, i.e., with slack $4 \cdot (\theta+1) \cdot (S + 1)$ for a given parameter $S\geq 1$ and with a color space $\mathcal{C}$ of size $|\mathcal{C}|=C$. To solve the given $P_D(4 \cdot (\theta+1) \cdot (S + 1), C)$ instance, we assume that we are given a list arbdefective coloring algorithm that solves instances of the form $P_A(S, C)$ with slack $S$ and the same color space $\mathcal{C}$. Since the list defective problem that we need to solve has slack $4 \cdot (\theta+1) \cdot (S + 1)$, we have color lists $L_v$ that satisfy the following condition:
\begin{align}\label{eq:InitialSlackCondition}
  \sum_{x \in L_v} (d_v(x) + 1) > 4 \cdot (\theta+1) \cdot (S + 1) \cdot \deg(v).
\end{align}
The algorithm for a given graph $G=(V,E)$ with neighborhood independence at most $\theta$ now proceeds as follows. We emphasize that we subsequently assume $d_v(x) \leq \Delta-1$, as otherwise, node $v$ can simply color itself with this high-defect color, reduce the defect for this color among its neighbors by one, and pretend that $v$ does not exist in subsequent steps.

\begin{enumerate}
    \item Each node $v\in V$ defines a new defect function $d_v':L_v\to \mathbb{N}_0$ as follows.
    \begin{align} \label{eq:DefectDPrime}
        \forall x\in L_v\,:\, d_v'(x) :=
        \max\left\{1,2^{\left\lfloor \log_2\frac{d_v(x) + 1}{2(\theta+1)} \right\rfloor}\right\} - 1.
    \end{align}
    Note that this new defect function assigns non-negative integer defects to all the colors such that for all $v$ and $x\in L_v$, $d_v'(x) + 1$ is an integer power of $2$. Since $d_v(x)\leq \Delta-1$ and $\theta\geq 1$, we have
    \[
        \frac{d_v(x)+1}{2(\theta+1)} \leq \frac{\Delta}{4} < 2^{\lceil\log\Delta\rceil-1},
    \]
    which implies $d'_v(x)+1\in \{2^0, 2^1, \dots, 2^{\lceil\log\Delta\rceil-1}\}$. For every $y>0$, we have
    $\max\{1,2^{\lfloor\log_2 y\rfloor}\}>y/2$. Hence, by \eqref{eq:InitialSlackCondition},
    \begin{align}\label{eq:RemainingSlack}
        \sum_{x \in L_v} (d_v'(x) + 1)
        &> \frac{1}{4(\theta+1)}\sum_{x\in L_v}(d_v(x)+1)
        > (S + 1) \cdot \deg(v).
    \end{align}

  \item Use \Cref{lemm:defColorBlackBox}  to compute a defective coloring with colors $\mathcal{Q}=\set{1,\dots,Q}$ such that $Q = O(\log^2 \Delta)$ and each node $v$ has at most $\deg(v)/ \lceil \log \Delta \rceil$ neighbors of the same color. This partitions the nodes $V$ of $G$ into $V_1, \dots, V_Q$. In the following, to distinguish the colors in $\mathcal{C}$ that we need to assign to the nodes and the colors in $\mathcal{Q}$ of this defective coloring, we refer to the former colors as $\mathcal{C}$-colors and to the latter colors as $\mathcal{Q}$-colors. For $q \in \{1, ..., Q\}$, we use $G_q=G[V_q]$ to denote the subgraph induced by the nodes colored with $\mathcal{Q}$-color $q$.

  \item The algorithm now consists of $\lceil\log\Delta\rceil$ phases $i = \lceil \log \Delta \rceil -1, \ldots, 0$ (in descending order). Each phase consists of $Q$ steps, where we iterate over the $\mathcal{Q}$-colors of the defective coloring. In step $s\in \mathcal{Q}$ of each phase $i$, we color a subset of the nodes in $V_s$ (details are specified below).  For each phase $i$, each node $v\in V$, and each $\mathcal{C}$-color $x\in L_v$ we define $a_v(x, i)$ as follows. Assume that $v\in V_s$. Then $a_v(x, i)$ denotes the number of neighbors of $v$ that are colored with $\mathcal{C}$-color $x \in L_v$ before step $s$ of phase $i$ starts.

  \item For a node $v\in V_s$ (for $s\in \mathcal{Q}$) and a phase $i$, we define
    \[
      \gamma_{v, i} := \sum_{x \in L_{v, i}}\big(d_v'(x) + 1 - a_v(x, i)\big),\text{ where }
      L_{v, i} := \{ x \in L_v | d_{v}'(x) + 1 = 2^i \}.
    \]
  Since $d'_v(x)+1\in \{2^0, \dots, 2^{\lceil\log\Delta\rceil-1}\}$, the sets $L_{v,0}, \dots, L_{v,\lceil\log\Delta\rceil-1}$ partition $L_v$. A node $v\in V_s$ is defined as \textit{active} in phase $i$ if it is still uncolored at the beginning of the phase and if $\gamma_{v, i} > S \cdot \deg_{G_s}(v)$, where $\deg_{G_s}(v)$ denotes the number of neighbors of $v$ in $V_s$.

  \item The algorithm iterates through phases $i = \lceil \log \Delta \rceil - 1, \ldots, 0$. In each phase, we further iterate through the $Q$ $\mathcal{Q}$-colors, where in step $s\in\set{1,\dots,Q}$ of phase $i$, we color all the nodes in $V_s$ that are \emph{active} in phase $i$ by using the given arbdefective coloring algorithm $P_A(S, C)$. Each such node $v$ uses the color list $L_{v,i}$ restricted to the colors that have non-negative defect when considering already colored neighbors from earlier phases. We call this list
    \[
      L^+_{v,i}:=\{x\in L_{v,i}\mid a_v(x,i)\leq d'_v(x)\}.
    \]
    The allowed arbdefect of $v$ for color $x\in L^+_{v,i}$ is defined as $\gamma_{v, i}(x) := d_v'(x) - a_v(x, i)$, which is non-negative by the definition of $L^+_{v,i}$. Moreover, every discarded color $x\in L_{v,i}\setminus L^+_{v,i}$ satisfies $d'_v(x)+1-a_v(x,i)\leq 0$. Therefore,
    \begin{align*}
      \sum_{x\in L^+_{v,i}}\big(\gamma_{v,i}(x)+1\big)
      &\geq \sum_{x\in L_{v,i}}\big(d'_v(x)+1-a_v(x,i)\big) \\
      &=\gamma_{v,i}>S\cdot\deg_{G_s}(v).
    \end{align*}
    Since the degree of $v$ in the subgraph induced by the active nodes of $V_s$ is at most $\deg_{G_s}(v)$, this is a legal $P_A(S,C)$ instance with slack at least $S$.
\end{enumerate}

For the analysis of our algorithm, we also define an orientation of all monochromatic edges in $G$ as follows. Edges between nodes that are participating in the same instance of $P_A(S, C)$ are oriented by the given arbdefective coloring algorithm. All other monochromatic edges between nodes colored in the phases are oriented from the node that is colored later to the node that is colored earlier.

\subparagraph{Algorithm Analysis.}  To show that the algorithm indeed solves the given list defective coloring instance $P_D(4 \cdot (\theta+1) \cdot (S + 1), C)$, we have to show two crucial statements. First, we have to show that every node eventually becomes \textit{active} and is therefore colored (\Cref{lemm:WillBeColored}). Furthermore, we need to demonstrate that after termination, nodes with $\mathcal{C}$-color $x$ have at most $d_v(x)$ neighbors of that color in $G$ (\Cref{lemm:Neighbors}).
\begin{lemma}
    \label{lemm:WillBeColored}
    The above algorithm colors every node $v\in V$. Furthermore, if $v$ is colored with color $x\in L_v$, then at most $d'_v(x)$ outneighbors of $v$ are colored with color $x$.
\end{lemma}
\begin{proof}
  We assume that $v\in V_s$ for $s\in \set{1,\dots,Q}$. To show that $v$ gets colored, we have to show that $v$ is active in some phase $i$ and thus that there exists an $i$ for which $\gamma_{v,i}>S\cdot \deg_{G_s}(v)$. For the sake of contradiction, suppose that $v$ remains uncolored throughout the algorithm. Then $v$ is uncolored at the beginning of every phase $i\in\{0,\dots,\lceil\log\Delta\rceil-1\}$, and because $v$ never becomes active, we have $\gamma_{v,i}\leq S\cdot \deg_{G_s}(v)$ for all phases $i$. Let $i_v$ be a phase $i$ that maximizes $\gamma_{v,i}$. We have
    \begin{align*}
        \gamma_{v, i_v} &\geq \frac{\sum_{i = 0}^{\lceil \log \Delta \rceil - 1} \gamma_{v, i}}{\lceil \log \Delta \rceil} \\
        &= \frac{\sum_{x \in L_v} (d_v'(x) + 1) - \sum_{i = 0}^{\lceil \log \Delta \rceil - 1} \sum_{x \in L_{v, i}} a_v(x, i)}{\lceil \log \Delta \rceil} \\
        &> \frac{ (S + 1) \cdot \deg(v) - \deg(v)}{\lceil \log \Delta \rceil} \\
        &\geq S \cdot \deg_{G_s}(v)
    \end{align*}
    The third step follows from \cref{eq:RemainingSlack} and the fact that the $L_{v, i}$ form a partition of $L_v$ and thus $\sum_i \sum_{x \in L_{v, i}} a_v(x, i) \leq \deg(v)$. The last step holds by the properties of the defective $Q$-coloring that the algorithm computes in step $2$ (using \Cref{lemm:defColorBlackBox}), we have $\deg_{G_s}(v) \leq \deg(v)/\lceil \log \Delta \rceil$. We thus get a contradiction to the assumption that $\gamma_{v,i}\leq S\cdot \deg_{G_s}(v)$ for all phases $i$, which proves that $v$ must become active and get colored in some phase.

    To prove the claimed upper bound on the number of outneighbors with the same color, assume that a node $v\in V_s$ gets colored with color $x\in L^+_{v,i}$ in step $s$ of some phase $i$. The set of outneighbors of $v$ that are colored with color $x$ consists of all neighbors of $v$ that are colored with color $x$ before step $s$ of phase $i$ and of the color-$x$ outneighbors $v$ gets in the instance of $P_A(S,C)$ in which $v$ gets colored. The edges to all neighbors that receive color $x$ in later steps and phases are oriented towards $v$. The number of neighbors of $v$ colored with color $x$ prior to step $s$ of phase $i$ is exactly $a_v(x,i)$. Further, because in the arbdefective coloring instance $P_A(S,C)$ that colors $v$, $v$ uses defect $\gamma_{v, i}(x)=d_v'(x)-a_v(x,i)$, $v$ gets at most $d_v'(x)-a_v(x,i)$ additional outneighbors of color $x$ in this step.
\end{proof}

\begin{lemma}
    \label{lemm:Neighbors}
    If the above algorithm colors a node $v\in V$ with color $x\in L_v$, then $v$ has at most $d_v(x)$ neighbors of color $x$ in $G$.
\end{lemma}
\begin{proof}

    By \Cref{lemm:WillBeColored}, node $v$ has at most $d_v'(x)$ outneighbors of color $x$. It therefore remains to count the maximum possible number of inneighbors of color $x$ that $v$ can get. Let $A$ be the set of inneighbors of $v$ that are colored with color $x$. For any $u\in A$, we know that $u$ is either colored in the same instance of the arbdefective coloring algorithm as $v$ or $u$ is colored after $v$. In both cases, $u$ is colored in the same phase as $v$ or in a later phase. This implies that $d'_u(x)\leq d'_v(x)$.

    We first consider the case $d'_v(x)=0$. If $A$ contained a node $u$, then $u$ would have $v$ as a color-$x$ outneighbor. On the other hand, $d'_u(x)\leq d'_v(x)=0$, and \Cref{lemm:WillBeColored} guarantees that $u$ has no color-$x$ outneighbor, which is a contradiction. Hence, $A=\emptyset$. Node $v$ also has no color-$x$ outneighbor, and therefore it has no neighbor of color $x$, as required.

    It remains to consider the case $d'_v(x)>0$. Every color-$x$ inneighbor of $v$ has at most $d'_v(x)$ outneighbors of color $x$. One of those outneighbors is $v$. Every node in $A$ thus has at most $d'_v(x)-1$ outneighbors in $A$ (as $v$ itself is not in $A$). Consequently, every induced subgraph $F$ of $G[A]$ inherits an orientation with maximum outdegree at most $d'_v(x)-1$. Hence, $|E(F)|\leq (d'_v(x)-1)|V(F)|$, so the average degree of $F$ is at most $2d'_v(x)-2$. Therefore, $F$ contains a vertex of degree at most $2d'_v(x)-2$, and $G[A]$ is $(2d'_v(x)-2)$-degenerate. It can thus be properly colored with $2d'_v(x)-1\leq 2d_v'(x)$ colors. Because the neighborhood independence of $G$ (and thus also of $G[A]$) is at most $\theta$, in every proper coloring of $G[A]$, at most $\theta$ nodes can have the same color. We can therefore conclude that $|A|\leq 2\theta\cdot d'_v(x)$. Since there are at most $d'_v(x)$ outneighbors of $v$ that can have color $x$, $v$ has at most $2\theta\cdot d'_v(x) + d'_v(x)\leq 2(\theta+1)\cdot d'_v(x)$ neighbors of color $x$.

    Because $d'_v(x)>0$, the maximum in \eqref{eq:DefectDPrime} is attained by its second argument. Thus,
    \[
        d'_v(x)+1\leq \frac{d_v(x)+1}{2(\theta+1)},
    \]
    and consequently
    \[
        2(\theta+1)d'_v(x)
        \leq d_v(x)+1-2(\theta+1)
        \leq d_v(x).
    \]
    This proves the claim.
\end{proof}

\restateforth*
\begin{proof}
  As proven in \Cref{lemm:WillBeColored,lemm:Neighbors}, the above algorithm correctly solves list defective coloring instances of the form $P_D(4\cdot(\theta+1)\cdot (S+1), C)$. Note that we have $S\geq 1$ and $\theta\geq 1$ and therefore $4(\theta+1)(S+1)\leq 16\theta S$. Our algorithm therefore in particular correctly solves list defective coloring instances of the form $P_D(16\theta S, C)$. Precoloring large defects in Step 1 takes $1$ round of communication to notify neighbors. The number of phases is $O(\log \Delta)$ and in each phase, we have $Q=O(\log^2\Delta)$ steps. The number of $P_A(S,C)$ instances that the algorithm needs to solve is therefore $O(\log^3\Delta)$. Finally, the algorithm applies the algorithm of \Cref{lemm:defColorBlackBox} once in order to compute the defective $Q$-coloring. With the initial proper $q$-coloring, this can be done in $O(\log^*q)$ rounds. This proves the claim on the round complexity $T_D(16\theta S, C)$.
\end{proof}
\newpage

\section{Analysis of the Recursive List Arbdefective Coloring Algorithm}
\label{sec:MainSection}

\Cref{lemma:SlackGeneration}, \Cref{lemm:colorSpaceReduction}, and \Cref{thm:ArbToDef} provide multiple ways to recursively express the complexity of list (arb)defective coloring instances. In this section, we combine those building blocks and choose appropriate parameters to obtain one recursive algorithm to solve list arbdefective coloring. The general structure of this recursive algorithm has already been discussed in \Cref{sec:recsubproblems} and is depicted in \Cref{fig:boundedNeighScheme}. While our main objective is to solve a $P_A(1,C)$ instance, our main recursive algorithm solves $P_A(2,C)$ instances. A $P_A(1,C)$-algorithm can then be obtained at an additional $O(\log\Delta)$ factor by using the first bound of \Cref{lemma:SlackGeneration}. Before discussing and analyzing the recursive algorithm in detail, we also provide an algorithm for the base case of the recursion, that is, we provide an algorithm that can solve instances $P_A(1,C)$ directly. For this, we use an algorithm that is provided in \cite{FK23arxiv} (which is the arxiv version of \cite{FK23}). We obtain the following lemma.
\begin{lemma}[\cite{FK23arxiv}]\label{lemm:DefColoringBaseCase}
   There is a deterministic \LOCAL algorithm to solve list arbdefective coloring instances $P_A(1,C)$ on any graph $G$ that is equipped with an initial proper $q$-coloring in time
    \begin{align*}
        T_A(1, C) = O(\sqrt{C} \cdot \log^2 \Delta \cdot\sqrt{\log\Delta + \log\log C}\cdot \log^{3/2} \log \Delta + \log^* q).
    \end{align*}
\end{lemma}
\begin{proof}
  We first compute a proper $O(\Delta^2)$-coloring of $G$ in $O(\log^* q)$ rounds by using a standard deterministic coloring algorithm of \cite{Linial1987}. In Theorem 1.3 of \cite{FK23arxiv}, it is shown that for real parameters $\nu\geq 0$ and $\kappa>0$, there is a deterministic \LOCAL algorithm to solve list arbdefective coloring instances with slack $1$ in time $O\big(\Lambda^{\nu/(1+\nu)}\cdot \kappa^{1/(1+\nu)}\cdot\log\Delta \cdot T_{\nu,\kappa}^{O} + \log^* q\big)$, where $\Lambda$ is the maximum list size and $T_{\nu,\kappa}^O$ is the time to solve a so-called \emph{oriented list defective} coloring instance with parameters $\nu$ and $\kappa$.\footnote{Oriented list defective coloring is a variant of list defective coloring for directed graphs and where the defect is only measured w.r.t.\ the outneighbors of a node.} We apply this theorem with $\nu=1$. Because all lists are contained in a color space of size $C$, we have $\Lambda\leq C$. Theorem 1.1 in \cite{FK23arxiv} states that if the graph is equipped with a proper $O(\Delta^2)$-coloring, the required oriented instances can be solved in $T_{1,\kappa}^O=O(\log\Delta)$ rounds. In the notation of that theorem, we may choose
  \[
    \kappa:=\left\lceil c(\log\Delta+\log\log C)\log^3\log\Delta\right\rceil>0
  \]
  for a sufficiently large absolute constant $c$. Substituting $\nu=1$, $\Lambda\leq C$, this value of $\kappa$, and $T_{1,\kappa}^O=O(\log\Delta)$ into the bound of Theorem 1.3 gives the claimed complexity.
\end{proof}

We are now ready to prove \Cref{thm:mainNeighborhood}, our main result. Instead of proving a bound on $(\Delta+1)$-coloring as it is done in \Cref{thm:mainNeighborhood}, the following lemma proves the upper bound directly for list arbdefective coloring instances with slack $1$. Note that \Cref{thm:mainNeighborhood} directly follows from \Cref{lemm:mainNeighborhood} because $(\Delta+1)$-coloring is a special case of list arbdefective coloring with slack $1$. Also, in the case of $(\Delta+1)$-coloring, the size of the color space is $C=O(\Delta)$.

\begin{lemma}
  \label{lemm:mainNeighborhood}
  In $n$-node graphs of maximum degree $\Delta$ and neighborhood independence $\theta$, we have the following two bounds on $T_A(1,C)$.
  \begin{align}
    T_A(1,C) &= (\theta\log\Delta)^{O\left(1 + \frac{\log\log C}{\max \{1, \log\log\log \Delta - \log \log \theta \}}\right)} + O(\log^* n),\label{eq:mainBound1}\\
    T_A(1,C) &= O\big(\theta^2\cdot C^{1/4} \cdot \log^6 \Delta
                   \cdot\sqrt{\log\Delta + \log\log C}\cdot \log^{3/2} \log \Delta + \log^* n\big).\label{eq:mainBound2}
  \end{align}
\end{lemma}
\begin{proof}
  In the following proof, we assume that we are given a proper $q = O(\Delta^2)$-coloring of $G$ (which can be computed in $O(\log^* n)$ rounds by using a deterministic algorithm of \cite{Linial1987}). As discussed above, we provide a recursive solution for $P_A(2,C)$ and we then use the first bound of \Cref{lemma:SlackGeneration} to obtain a bound on $T_A(1,C)$. Starting from $P_A(2,C)$, the first task is to generate slack $S := 64 \cdot \theta$, which can be achieved by using the second bound of \Cref{lemma:SlackGeneration}. Then, we solve this arbdefective coloring problem using color space reduction as of \Cref{lemm:colorSpaceReduction} with parameters $\sigma := S/2=32\theta$ and $p := C^{\eps}$ for some appropriate value of $\eps>0$.\footnote{For simplicity, we assume that $C^{\eps}$ and $C^{1-\eps}$ are integers.} We then utilize \Cref{thm:ArbToDef} with parameter $2$ to solve the list defective coloring part of the color space reduction. At the level of subproblem types and multiplicities, the complete pipeline is, as depcited in \Cref{fig:boundedNeighScheme},
  \[
    \begin{aligned}
      P_A(2,C)
      &\longrightarrow O(\theta^2)\text{ instances of }P_A(64\theta,C)\\
      &\longrightarrow P_D(32\theta,C^\eps)
        +P_A(2,C^{1-\eps})\\
      &\longrightarrow O(\log^3\Delta)P_A(2,C^\eps)
        +P_A(2,C^{1-\eps}).
    \end{aligned}
  \]
  The transformations in the last two lines apply to each of the $O(\theta^2)$ instances generated in the first line. The term $P_A(2,C^{1-\eps})$ represents the $C^\eps$ induced subinstances produced by color space reduction. They run in parallel because each vertex is assigned to exactly one subinstance and the subinstances use disjoint palettes. The pipeline results in the following recursive formula for $T_A(2,C)$.
  \begin{align}
    T_A(2, C) &\leq O(\theta^2) \cdot T_A \left(64 \cdot \theta, C \right) + O(\log^* \Delta) \notag \\
              &\leq O(\theta^2) \cdot \left( T_D\big(32 \theta, C^\eps\big) + T_A\big(2, C^{1 - \eps}\big)\right) + O(\log^* \Delta)  \notag \\
              &\leq O(\theta^2) \cdot \left( \log^3 \Delta \cdot T_A\big(2, C^{\eps}\big) + T_A\big(2, C^{1 - \eps}\big)\right)  \label{eq:splitting}
  \end{align}
  Note that when applying \Cref{thm:ArbToDef} in the last inequality, we obtain an additional additive $O(\theta^2\cdot\log^*\Delta)$ term. This term and the $O(\log^*\Delta)$ term from the first two lines are however clearly asymptotically dominated by $O(\theta^2\cdot\log^3\Delta\cdot T_A(2, C^{\eps}))$. In order to obtain \eqref{eq:mainBound1}, we choose $\varepsilon$ to balance the recursion costs as in \Cref{lemma:techlemma1Corrected}, which leads to
  \begin{equation}\label{eq:mainBound1Slack}
    T_A(2,C) \leq (\theta\log\Delta)^{O\left(1 + \frac{\log\log C}{\max\{1, \log\log\log\Delta - \log\log\theta\}}\right)}.
  \end{equation}
  As the necessary algebraic manipulations are a bit lengthy, the proof of this appears in \Cref{cor:optimizedTechnicalBound} in \Cref{sec:technicalProofs}. When combining \eqref{eq:mainBound1Slack} with the first bound of \Cref{lemma:SlackGeneration}, we get
  \begin{align*}
    T_A(1,C) &\leq O(\log\Delta)\cdot (\theta\log\Delta)^{O\left(1 + \frac{\log\log C}{\max\{1, \log\log\log\Delta - \log\log\theta\}}\right)} + O(\log\Delta\cdot\log^*\Delta) \\
    &= (\theta\log\Delta)^{O\left(1 + \frac{\log\log C}{\max\{1, \log\log\log\Delta - \log\log\theta\}}\right)}.
  \end{align*}
  Together with the $O(\log^* n)$ rounds to compute the initial proper $O(\Delta^2)$-coloring, this proves the first bound of the lemma.
  To obtain the second bound \eqref{eq:mainBound2} of the lemma, we set $\eps=1/2$. We then only apply the recursive definition of \eqref{eq:splitting} once before we directly apply the bound of \Cref{lemm:DefColoringBaseCase}. We then have
  \begin{eqnarray*}
    T_A(2,C) & \leq  & O(\theta^2) \cdot \left( \log^3 \Delta \cdot T_A\big(2, C^{\eps}\big) + T_A\big(2, C^{1 - \eps}\big)\right)\\
             & \stackrel{(\eps=1/2)}{\leq}
                     & O(\theta^2\log^3\Delta)\cdot T_A\big(2,\sqrt{C}\big)\\
             & \leq &
                      O(\theta^2\log^3\Delta)\cdot O\big(C^{1/4} \cdot \log^2 \Delta \cdot\sqrt{\log\Delta + \log\log C}\cdot \log^{3/2} \log \Delta\big).
  \end{eqnarray*}
  Together with the bound $T_A(1,C)\leq O(\log\Delta)\cdot T_A(2,C) + O(\log\Delta \cdot\log^* \Delta)$ and the $O(\log^* n)$ time bound to compute the initial proper $O(\Delta^2)$-coloring, this directly implies the second bound \eqref{eq:mainBound2} of the lemma.
\end{proof}

We can now derive the main theorem by specializing the preceding lemma to
$(\Delta+1)$-coloring.

\restateMain*
\begin{proof}
  A $(\Delta+1)$-coloring instance is a list arbdefective coloring instance
  with slack $1$, zero allowed arbdefects, and color-space size
  $C=\Delta+1=O(\Delta)$. Substituting this value of $C$ into
  \eqref{eq:mainBound1} gives the first bound in the theorem.

  For the second bound, \eqref{eq:mainBound2} and $C=O(\Delta)$ give
  \[
    T_A(1,\Delta+1)
    =O\left(
      \theta^2\Delta^{1/4}\log^6\Delta
      \sqrt{\log\Delta+\log\log\Delta}\,
      \log^{3/2}\log\Delta
      +\log^*n
    \right).
  \]
  We have
  $\sqrt{\log\Delta+\log\log\Delta}=O(\sqrt{\log\Delta})$ and
  $(\log\log\Delta)^3=O(\log\Delta)$, which implies
  $\log^{3/2}\log\Delta=O(\sqrt{\log\Delta})$. Hence the product of the
  two displayed logarithmic factors is $O(\log\Delta)$, and the second
  bound simplifies to
  \[
    O\big(\theta^2\Delta^{1/4}\log^7\Delta+\log^*n\big).
  \]
  The logarithm convention from \Cref{sec:intro} absorbs bounded values of
  $\Delta$ into the implicit constants. Taking the better of the two bounds
  proves the theorem.
\end{proof}

\subsection{Technical Part of the Main Theorem Proof}
\label{sec:technicalProofs}

In this section we provide the technical calculations needed for the recursion analysis to obtain the main result \Cref{thm:mainNeighborhood}. More specifically, we prove \Cref{lemma:techlemma1Corrected}, whose simplified version, stated in \Cref{cor:optimizedTechnicalBound}, is used in the proof of \Cref{lemm:mainNeighborhood}.

\begin{lemma}[Optimized balanced recursion]
    \label{lemma:techlemma1Corrected}
    Assume that the recurrence in
    \eqref{eq:splitting} is available for every choice of
    $0<\eps<1$. Fix a constant $\alpha\geq 2$ such that
    \begin{align}\label{eq:splittingBalanced}
        T_A(2,C)
        \leq \alpha\theta^2\left(
            \log^3\Delta\cdot T_A(2,C^\eps)
            +T_A(2,C^{1-\eps})
        \right)
    \end{align}
    holds, and abbreviate
    \[
        A:=\alpha\theta^2
        \qquad\text{and}\qquad
        B:=\alpha\theta^2\log^3\Delta.
    \]
    Let $q>0$ be the unique solution of
    \begin{equation}\label{eq:criticalExponent}
        B^{-1/q}+A^{-1/q}=1,
    \end{equation}
    and set
    \begin{equation}\label{eq:balancedEpsilon}
        \eps:=B^{-1/q}
        \qquad\text{and hence}\qquad
        1-\eps=A^{-1/q}.
    \end{equation}
    Then, for every $C\geq 2$,
    \begin{align}\label{eq:correctedTechnicalBound}
        T_A(2,C)
        &\leq O(\log^4\Delta)\cdot B^{1+1/q}
        \max\left\{1,
            \left(\frac{\log C}{\log\log\Delta}\right)^{q+1}
        \right\}.
    \end{align}
\end{lemma}
\begin{proof}
    We first explain the parameter choice by considering the homogeneous
    part of the recurrence in the logarithm of the color-space size. After writing $x:=\log C$,
    the recurrence \eqref{eq:splittingBalanced} has the form
    \begin{equation}\label{eq:additiveColorRecursion}
        T(x)\leq B\,T(\eps x)+A\,T((1-\eps)x).
    \end{equation}
    For a prospective bound of the form $Kx^{q+1}$, the two recursive terms give
    \[
        BK(\eps x)^{q+1}+AK((1-\eps)x)^{q+1}
        =Kx^{q+1}
        \underbrace{\left(B\eps^{q+1}
        +A(1-\eps)^{q+1}\right)}_{=:M_q(\eps)}.
    \]
    Hence, the prospective bound is preserved by the recurrence if
    $M_q(\eps)\leq 1$.

    For fixed $q>0$, the function $M_q$ is strictly convex, and
    \[
        M_q'(\eps)
        =(q+1)\left(B\eps^q-A(1-\eps)^q\right),
    \]
    implies that its unique minimizer satisfies
    $B\eps^q=A(1-\eps)^q$. If $K_q$ denotes this common value, then at
    the minimizer
    \[
        M_q(\eps)=K_q\eps+K_q(1-\eps)=K_q.
    \]
    We choose the critical exponent $q$ for which this minimum is exactly
    $1$. Equivalently, we seek $q>0$ and $0<\eps<1$ such that
    \[
        B\eps^q=A(1-\eps)^q=1.
    \]
    These identities force $\eps=B^{-1/q}$ and
    $1-\eps=A^{-1/q}$. They are compatible precisely when
    \eqref{eq:criticalExponent} holds. The left-hand side of that
    equation is continuous and strictly increasing in $q$, from $0$ as
    $q\to0^+$ to $2$ as $q\to\infty$. Thus, $q$ exists and is unique.
    Moreover, $B\geq A\geq2$ implies
    $B^{-1}+A^{-1}\leq1$, so monotonicity gives $q\geq1$. Finally,
    \eqref{eq:balancedEpsilon} gives
    \begin{equation}\label{eq:balancedMultiplier}
        B\eps^{q+1}+A(1-\eps)^{q+1}
        =\eps+(1-\eps)=1.
    \end{equation}
    This is the optimized form of balancing the two recursive costs.
    Notice also that $\eps\leq 1-\eps$ because $B\geq A$.

    We next justify the bound, including the base cases. Put
    $x_0:=\log\log\Delta$. If $x\leq x_0$, or equivalently
    $C\leq\log\Delta$, then \Cref{lemm:DefColoringBaseCase} and
    $T_A(2,C)\leq T_A(1,C)$ give the following simplification. For
    $\Delta\geq 16$, we have
    \[
        \sqrt C\leq\sqrt{\log\Delta},\qquad
        \sqrt{\log\Delta+\log\log C}=O(\sqrt{\log\Delta}),
        \qquad
        \log^{3/2}\log\Delta=O(\log\Delta).
    \]
    Together with $q=O(\Delta^2)$ and thus $\log^*q=O(\log\Delta)$,
    these inequalities show that the bound in
    \Cref{lemm:DefColoringBaseCase} is $O(\log^4\Delta)$. Note that \Cref{lemm:DefColoringBaseCase} also handles $\Delta < 16$. Thus we can write that for all $C\leq\log\Delta$,
    \begin{equation}\label{eq:balancedBaseCase}
        T_A(2,C)\leq D\log^4\Delta
    \end{equation}
    for an absolute constant $D$ or equivalently $T(x)\leq D \log^4\Delta$ for all $x\leq x_0$.

    For $x>x_0$, we unroll the recurrence \eqref{eq:additiveColorRecursion} into a
    binary recursion tree. The root corresponds to the initial problem with
    parameter $x$. Each node with parameter $y>x_0$ branches into a left child
    with parameter $\eps y$ (multiplying the recursive cost by $B$) and a right
    child with parameter $(1-\eps)y$ (multiplying the recursive cost by $A$). We
    stop branching along any path as soon as its parameter is at most $x_0$;
    these terminal nodes form the set of leaves $\mathcal{Z}$. The tree is finite
    because both shrink factors $\eps$ and $1-\eps$ are strictly smaller than $1$.

    For a leaf $z\in\mathcal{Z}$, let $k$ and $m$ be the number of left and right
    steps on the path from the root to $z$, respectively. The parameter at leaf $z$
    is $x_z:=x\cdot\eps^k(1-\eps)^m$, so that $s_z:=x_z/x=\eps^k(1-\eps)^m$ is the
    total shrink factor along the path. The cumulative coefficient attached to
    $T(x_z)$ in the unrolled recurrence is the weight $w_z:=B^k A^m$. Fully
    unrolling the recurrence down to the leaves gives
    $T(x)\leq\sum_{z\in\mathcal{Z}}w_z T(x_z)$.

    Because the parent of leaf $z$ had parameter strictly larger than $x_0$, and each
    step shrinks the parameter by a factor of at least $\min\{\eps,1-\eps\}=\eps$,
    we have
    \begin{equation}\label{eq:leafParameterRange}
        \eps x_0<x_z\leq x_0.
    \end{equation}

    Assign probability $B\eps^p=\eps$ to the left child and probability
    $A(1-\eps)^p=1-\eps$ to the right child. These probabilities sum to
    $1$ by \eqref{eq:balancedMultiplier}. The probability $\pi_z$ of
    the root-to-leaf path ending at $z$ therefore satisfies
    \[
        \pi_z= \eps^k (1-\eps)^m = B^kA^m \eps^{pk} (1-\eps)^{pm} = w_zs_z^p,
        \qquad\text{and}\qquad
        \sum_{z\in\mathcal{Z}}\pi_z=1.
    \]
    It follows from \eqref{eq:leafParameterRange} that
    \begin{align*}
        \sum_{z\in\mathcal{Z}} w_z
        &=\sum_{z\in\mathcal{Z}}\pi_z\left(\frac{x}{x_z}\right)^p
        \leq\left(\frac{x}{\eps x_0}\right)^p \sum_{z\in\mathcal{Z}}\pi_z
        =\left(\frac{x}{\eps x_0}\right)^p.
    \end{align*}
    Since every leaf $z\in\mathcal{Z}$ satisfies $x_z\leq x_0$, we apply the base-case
    bound \eqref{eq:balancedBaseCase}, which gives $T(x_z)\leq D\log^4\Delta$ for all $z\in\mathcal{Z}$.
    Substituting this and the bound on $\sum_{z\in\mathcal{Z}} w_z$ into
    $T(x)\leq\sum_{z\in\mathcal{Z}} w_z T(x_z)$ gives
    \begin{align*}
        T_A(2,C) = T(x)
        &\leq \sum_{z\in\mathcal{Z}} w_z \cdot D\log^4\Delta
        = D\log^4\Delta \sum_{z\in\mathcal{Z}} w_z \\
        &\leq D\log^4\Delta \left(\frac{x}{\eps x_0}\right)^p
        = D\log^4\Delta \left(\frac{\log C}{\eps\log\log\Delta}\right)^p.
    \end{align*}
    Finally, $\eps^{-p}=(B^{-1/q})^{-(q+1)}=B^{(q+1)/q}=B^{1+1/q}$. Together with the base case
    itself ($C\leq\log\Delta$), this proves \eqref{eq:correctedTechnicalBound}.
\end{proof}

\begin{lemma}[Simplification of the optimized recursion bound]
    \label{lemma:techlemma1Simplified}
    Under the assumptions and with the notation of
    \Cref{lemma:techlemma1Corrected}, for arbitrary $\theta\geq 1$ and
    every $C\geq 2$, define
    \begin{equation}\label{eq:lambdaForSimplifiedBound}
        \lambda
        :=\log\left(1+\frac{\log B}{\log A}\right)
        =\log\left(
            2+\frac{3\log\log\Delta}{\log(\alpha\theta^2)}
        \right).
    \end{equation}
    We then have
    \begin{equation}\label{eq:criticalExponentAsymptotic}
        q=\Theta\left(\frac{\log B}{\lambda}\right)
    \end{equation}
    and consequently
    \begin{align}\label{eq:simplifiedTechnicalBound}
        T_A(2,C)
        \leq B^{
            O\left(
                1+\frac{\log\log C}{\lambda}
            \right)}.
    \end{align}
\end{lemma}
\begin{proof}
    We first estimate $q$. For this calculation only, write
    $a:=\ln A$, $b:=\ln B$, $r:=b/a\geq 1$, and $u:=b/q$.
    Equation \eqref{eq:criticalExponent} becomes
    \[
        e^{-u}+e^{-u/r}=1.
    \]
    Using $r \geq 1$, this implies that
    \begin{align*}
        1 = e^{-u}+e^{-u/r} \leq e^{-u/r}+e^{-u/r} = 2e^{-u/r} \text{ and } 1 \geq 2e^{-u}
    \end{align*}
    And thus $\ln 2 \leq u \leq r \cdot \ln 2$ and simplified $0 < u/r \leq \ln 2$. Using
    $y/2\leq 1-e^{-y}\leq y$ for $0\leq y\leq\ln 2$ on this range of $u/r$ therefore yields
    \[
        \frac{u}{2r}\leq e^{-u}\leq\frac{u}{r},
        \qquad\text{and hence}\qquad
        r\leq ue^u\leq 2r.
    \]
    Consequently, $u=\Theta(\ln(1+r))$, which shows
    \begin{align*}
        u = \Theta(\log (1+r)) = \Theta\left(\log \left(1 + \frac{\log B}{\log A} \right) \right) = \Theta(\lambda)
    \end{align*}
    and thus yields \eqref{eq:criticalExponentAsymptotic} with absolute implicit
    constants.

    Since $q\geq 1$, we have $B^{1+1/q}\leq B^2$. Moreover,
    $B\geq\log^3\Delta$ implies
    \[
        O(\log^4\Delta)B^{1+1/q}=B^{O(1)}.
    \]
    If $\log C\leq\log\log\Delta$, the maximum in
    \eqref{eq:correctedTechnicalBound} equals $1$, and the claimed
    bound follows. Otherwise, \eqref{eq:criticalExponentAsymptotic}
    gives $q+1=O(\log B/\lambda)$, and hence
    \begin{align*}
        \left(\frac{\log C}{\log\log\Delta}\right)^{q+1}
        &=B^{
            (q+1)\frac{\log(\log C/\log\log\Delta)}{\log B}}
        \leq B^{O(\log\log C/\lambda)}.
    \end{align*}
    Substitution into \eqref{eq:correctedTechnicalBound} proves
    \eqref{eq:simplifiedTechnicalBound}.
\end{proof}

\begin{corollary}
    \label{cor:optimizedTechnicalBound}
    Under the assumptions of \Cref{lemma:techlemma1Corrected}, assume
    $1\leq\theta\leq\Delta$. Then, for every
    $C\geq 2$,
    \begin{align}\label{eq:optimizedTechnicalBoundAllTheta}
        T_A(2,C)
        \leq\displaystyle
            \left(\theta\log\Delta\right)^{
                O\left(
                    1+
                    \dfrac{\log\log C}{
                        \max \{1, \log\log\log \Delta - \log \log \theta \}
                    }
                \right)}
    \end{align}
\end{corollary}
\begin{proof}
    Let
    \[
        D:=\max\{1,\log\log\log\Delta-\log\log\theta\}.
    \]
    We first compare $\lambda$ from \eqref{eq:lambdaForSimplifiedBound} with
    $D$. Since $\alpha\geq2$ is a constant, we have
    \[
        \log(\alpha\theta^2)=\Theta(1+\log\theta).
    \]
    If $\log\theta\geq\log\log\Delta$, then $D=1$, while
    \eqref{eq:lambdaForSimplifiedBound} gives $\lambda=\Theta(1)$. If
    $\log\theta<\log\log\Delta$, then
    \begin{align*}
        \lambda
        &=\Theta\left(
            1+\log\frac{\log\log\Delta}{1+\log\theta}
          \right) \\
        &=\Theta\bigl(
            \max\{1,\log\log\log\Delta-\log\log\theta\}
          \bigr)
         =\Theta(D),
    \end{align*}
    where the logarithm convention handles $\theta=1$ and all bounded
    parameter values. Thus, in all cases, $\lambda=\Theta(D)$.

    It remains to replace the base $B$ in
    \eqref{eq:simplifiedTechnicalBound}. By definition,
    \[
        B=\alpha\theta^2\log^3\Delta
         \leq \alpha(\theta\log\Delta)^3
         =(\theta\log\Delta)^{O(1)},
    \]
    because $\theta\geq1$ and $\log\Delta\geq1$. Substituting
    $\lambda=\Theta(D)$ into \eqref{eq:simplifiedTechnicalBound} and then
    replacing $B$ by $(\theta\log\Delta)^{O(1)}$ gives exactly
    \eqref{eq:optimizedTechnicalBoundAllTheta}.
\end{proof}

\section{Lower Bound on Defective Edge Coloring of Hypergraphs}
\label{sec:lowerbound}

In this section we prove the lower bound on defective $2$-edge-coloring as stated in \Cref{thm:lowerbound}.

Our proof is based on a lower bound for defective $2$-vertex coloring of trees given in \cite{BalliuHLOS19}. Here, the authors introduce the \DCLOCAL\ model, which is an adaptation of the \LOCAL model.

\begin{definition}[\DCLOCAL$(k, c)$]\label{def:DC-LOCAL}
  The \DCLOCAL$(k,c)$ model for integer parameters $k\geq 1$, $c\geq 2$ is identical to \LOCAL model with the following adaptation. The nodes do not have unique identifiers. Instead, they are initially equipped with a proper distance-$k$ coloring with colors $\set{1,\dots,c}$.
\end{definition}

It is well-known that locally checkable labeling (LCL) problems (in bounded-degree graphs) that can be solved in $O(\log^* n)$ rounds in \LOCAL, can be solved deterministically in $k$ rounds in the \DCLOCAL$(k,c)$ model for some parameters $k=O(1)$ and any $c=O(1)$~\cite{chang16exponential}. It is further known that every LCL problem that can be deterministically solved in $o(\log n)$ rounds, can also be solved deterministically in $O(\log^* n)$ rounds in the \LOCAL model. In order to show that $(r(d-1)-d)$-defective edge coloring with $2$ colors requires $\Omega(\log n)$ rounds in $d$-regular, $r$-uniform hypergraphs (for any constants $d$, $r$), it is therefore equivalent to show that there is no \DCLOCAL$(k,c)$ for any $k=O(1)$ that solves the same problem in $O(1)$ rounds. As in \cite{BalliuHLOS19}, we prove this by a reduction that shows that given some $O(1)$-round \DCLOCAL$(k,c)$ algorithm to compute the required defective edge coloring implies an $O(\log^* n)$-round algorithm to compute a sinkless orientation in $2$-colored regular bipartite graphs, which is in contradiction to the sinkless orientation lower bound of \cite{ChangHLPU18}.

Before we get to the actual lower bound, we start with some terminology. First, we define the notion of a \emph{$(r,d)$-regular tree} as follows
\begin{definition}[$(r, d)$-regular tree]\label{def:rd-regular}
    We call a tree $T$ a $(r, d)$-regular tree for some integers $r,d > 1$ if in every even layer (assuming the root of the tree is in layer $0$), each non-leaf node has degree exactly $r$ and in every odd layer, each non-leaf node has degree exactly $d$.
\end{definition}
We further define the \emph{bipartite representation} of a hypergraph as follows.
\begin{definition}[bipartite hypergraph representation]\label{def:biprepresentation}
    Given a hypergraph $H=(V,E)$ with nodes $V$ and hyperedges $E\subseteq 2^V$, the \emph{bipartite representation}  of $H$ is a bipartite graph $B=(V_B, E_B)$ that contains one node for each node $v\in V$ of $H$ and for each edge $e\in E$ of $H$ (i.e., $V_B=V\cup E$) and where there is an edge $\set{v,e}\in E_B$ for $v\in V$ and $e\in E$ if and only if $v\in e$.
\end{definition}

The following lemma proves a structural property for $(r(d-1)-d)$-defective $2$-edge colorings of $d$-regular $r$-uniform hypergraphs.
  \begin{lemma}
    \label{lemma:regularTreeProperty}
    Let $H=(V,E)$ be a $d$-regular, $r$-uniform hypergraph for $d,r\geq 2$. If the edges $E$ are colored with two colors such that every edge has at most $(r(d-1)-d)$ neighboring edges of the same color, the following holds for any edge $e\in E$ and any integer $t\geq 0$. If in the bipartite representation of $H$, the $4t$-hop neighborhood of $e$ has no cycles and no leaves, then there exists a node (representing an edge $e'\in E$) at distance exactly $4t$ from $e$ such that $e$ and $e'$ have the same color.
\end{lemma}
\begin{proof}
  Let $e$ be the node in $H$ as depicted in the lemma and let $e$ be colored with $x \in \{black, white\}$. As described in the statement, the $4t$-hop view of $e$ is a complete $(r,d)$-regular tree. Consider this tree as a rooted tree, rooted at node $e$. We need to show that in this rooted tree, there is a node $e'$ that is colored with color $x$ in layer $4t$ (we say that $e$ itself is in layer $0$).

  As $e$ is the root of the tree, it has exactly $r \cdot (d-1)$ distance-$2$ neighbors, which are all in layer $2$ of the tree. The number of those nodes with color $\overline{x}$ (inverse color of $x$) has to be $\geq r \cdot (d-1)- (r(d-1)-d) = d \geq 2$. We can thus guarantee the existence of some node $e_2$ in distance exactly $2$ of $e$ with color $\overline{x}$. We will now use a similar argument to show that in the subtree of $e_2$, there is at least one node with color $x$ in distance $2$ from $e_2$. Node $e_2$ has $r-1$ children and thus $(r-1) \cdot (d-1)$ grandchildren (with respect to this tree rooted at $e$). If all grandchildren are colored with color $\overline{x}$, we run into a contradiction with the promise that at most $(r(d-1)-d)$ of the $2$-hop neighbors have the same color as $e_2$. Note that $r(d-1)-d < (r-1) \cdot (d-1)$. Hence, at least one of the grandchildren of $e_2$ must be colored with color $x$. This argument can be made for every node on an even layer (i.e., for every degree-$r$ node). Every such node must have at least one grandchild of the opposite color. This however implies that for every node at an even layer $2\ell$, there is at least one node on layer $2\ell+4$ (as long as $2\ell+4\leq 4t$) with the same color. This proves the claim of the lemma.
\end{proof}

\begin{lemma}
    \label{lemma:lower_reduction}
    Assume the existence of some algorithm $\calA$ that computes a $(r(d-1)-d)$-defective $2$-edge coloring of $r$-uniform hypergraphs of maximum degree $d$ in $k$ rounds in the \DCLOCAL$(k, (r\cdot d)^{k})$ model for some constant $k\geq 1$, then one can also compute sinkless orientations on $\Delta := r \cdot (d-1)^{2k} \cdot (r-1)^{2k-1}$-regular $2$-colored trees in $O(\log^* n)$ rounds in \LOCAL.
\end{lemma}
\begin{proof}
  Consider a $\Delta$-regular $2$-colored tree $B = (U \cup V, E)$ in the \LOCAL model, where $U$ and $V$ represent the set of nodes belonging to the two color classes. Now, we transform $B$ into a $(r,d)$-regular tree as follows. Each node $w \in U \cup V$ constructs a virtual tree rooted at $w$. This virtual tree is a $(r, d)$-regular tree as defined in \Cref{def:rd-regular} of height exactly $4k$. We denote this virtual tree by $T_{virt}(w)$. Note that each such tree has exactly $\Delta$ leaves.

  Consider two nodes $w_1\in U$ and $w_2\in V$ such that $\set{w_1,w_2}\in E$ and such that $w_2$ is the neighbor of $w_1$ on port $i$ and $w_1$ is the neighbor of $w_2$ on port $j$ (for $i,j\in \{1,\dots\Delta \})$. We then merge the $i$-th leaf of $T_{virt}(w_1)$ with the $j$-th leaf of $T_{virt}(w_2)$ into a merged node that we call $m_{w_1, w_2}$. This results in a virtual tree that is $(r,d)$-regular except for those merged nodes, which have degree $2$ (instead of $r$). In order to make the tree $(r,d)$-regular, we attach $r-2$ ``dummy nodes'' to each merged node. Note that this $(r,d)$-regular tree is the bipartite representation of some $r$-uniform hypergraph of maximum degree $d$. We can therefore apply our defective edge coloring algorithm $\calA$ to $2$-color the degree-$r$ nodes with $2$ colors such that every degree-$r$ node has at least $d$ distance-$2$ neighbors of the opposite color. In order to do this, we however first have to equip the simulated $(r,d)$-regular tree with a proper distance-$k$ coloring with $(r\cdot d)^{2k}$. We further want to construct this distance-$k$ so that for every $w\in U$, the root node of the tree $T_{virt}(w)$ is colored black and for every node $w\in V$, the root node of the tree $T_{virt}(w)$ is colored white. Note that the distance-$k$ neighborhoods of the roots of different virtual trees are disjoint and more than $k$ hops apart from each other. Each node $w\in U\cup V$ can therefore independently determine the distance-$k$ coloring of its $k$-hop neighborhood. As the maximum number of distance-$k$ neighbors in a $(r,d)$-regular tree is less than $(r\cdot d)^{2k}$, this partial proper distance-$k$ coloring of the virtual $(r,d)$-regular tree can then be completed in $O(\log^* n)$ rounds by a standard distributed coloring algorithm~(e.g., \cite{barenboim14distributed}). Also note that all the virtual trees are complete $(r,d)$-regular trees up to distance $k$ (in fact, up to distance $4k$). For such neighborhoods, there clearly must exist initial distance-$k$ colorings that force the root to be black and distance-$k$ colorings that force the root to be white (note that the distance-$2$ neighbors of the root nodes see the same topology). We can therefore force all the nodes on the $U$-side of the simulated $(r,d)$-regular tree to be colored black and all nodes on the $V$-side of the simulated $(r,d)$-regular tree to be colored white.

  By \Cref{lemma:regularTreeProperty}, the computed $2$-coloring of the degree-$r$ nodes has the following property. For every node $w\in U\cup V$, in the tree $T_{virt}(w)$, at least one node at distance $4k$ from the root has the same color as the root. The nodes at distance $4k$ from the root are the merged leaf nodes of the virtual trees. For each $w\in U$, we can therefore conclude that at least one of its merged leaf nodes is black and for each $w\in V$, we can conclude that at least one of its merged leaf nodes is white. We can use this property to compute a sinkless orientation of $2$-colored bipartite graph $B=(U\cup V, E)$ from which we started. Edges between degree-$\Delta$ and degree-$1$ nodes in $B$ are oriented towards the degree-$1$ node (degree-$1$ nodes are allowed to be sinks). Edges $\set{w_1,w_2}\in E$ with $w_1\in U$ and $w_2\in V$ and where $\deg_B(w_1)=\deg_B(w_2)=\Delta$ are oriented as follows. If the merged node $m_{w_1,w_2}$ is colored black, the edge is oriented from $w_1$ to $w_2$ and if $m_{w_2,w_1}$ is colored white, the edge is oriented from $w_2$ to $w_1$. Because $T_{virt}(w_1)$ is guaranteed to have at least one black leaf and $T_{virt}(w_2)$ is guaranteed to have at least one white leaf, this guarantees that every degree-$\Delta$ node of $B$ receives at least one outgoing edge (and thus is not a sink). We can therefore conclude that if the defective $2$-edge coloring algorithm $\calA$ exists as claimed, then a sinkless orientation of bipartite $2$-colored $\Delta$-regular trees can be computed in $O(\log^* n)$ rounds in \LOCAL.
\end{proof}

Based on \Cref{lemma:lower_reduction}, the proof of \Cref{thm:lowerbound} is now straightforward. For completeness, we first restate the theorem.

\restateLowerBound*

\begin{proof}
  For a $\Delta$-regular tree, we call an orientation of the edges a sinkless orientation if every degree-$\Delta$ node has at least one outgoing edge. From \cite{ChangHLPU18}, it is known that for $\Delta=O(1)$, deterministically computing a sinkless orientation in $2$-colored $\Delta$-regular trees requires $\Omega(\log n)$ rounds in the \LOCAL model. As discussed above, from the results in \cite{chang16exponential}, we further know that in bounded-degree graphs, any LCL problem (which includes defective coloring problems) can be solved deterministically in $o(\log n)$ rounds if and only if there exists a $k$-round $\DCLOCAL(k,c)$ for some constant $k\geq 1$ and any constant $c\geq1$ for the problem. \Cref{lemma:lower_reduction} therefore implies that if there exists an $o(\log n)$-round deterministic algorithm to compute an $(r(d-1)-d)$-defective $2$-edge coloring of $r$-uniform hypergraphs of maximum degree $d$, then there is an $O(\log^* n)$-round deterministic algorithm to compute a sinkless orientation in $2$-colored $\Delta$-regular trees. This is in contradiction to the lower bound of \cite{ChangHLPU18}.
\end{proof}
\newpage

\section{Conclusions}
In this paper, we showed that in graphs of bounded neighborhood independence $\theta=O(1)$, one can solve
$(\Delta+1)$-coloring in $(\log\Delta)^{O\big(\frac{\log\log\Delta}{\log\log\log\Delta}\big)}+O(\log^* n)$, and more generally for any $\theta\leq\Delta$ in $(\theta\log\Delta)^{O\left(1 + \frac{\log\log\Delta}{\max\{1, \log\log\log\Delta - \log\log\theta\}}\right)}+O(\log^* n)$ rounds. We also showed that it suffices to have $\theta\leq\Delta^{1/8-\eps}$ in order to get a $(\Delta+1)$-coloring algorithm that is polynomially faster than the current best algorithm for general graphs. One natural open question is whether we can improve the complexity for graphs with $\theta=O(1)$ or $\theta=\poly\log\Delta$ from quasipolynomial in $\log\Delta$ to $\poly\log\Delta$. It is known that this is possible for the special case of $(2\Delta-1)$-edge coloring~\cite{BalliuBKO22}. Our lower bound however shows that the approach of \cite{BalliuBKO22} already fails for edge coloring $3$-uniform hypergraphs with $3\Delta-2$ colors. Another natural open question is whether $o(\sqrt{\Delta})+O(\log^*n)$-round $(\Delta+1)$-coloring algorithms exist for general graphs. At least based on the current techniques, it seems that in order to understand this question, it might be necessary to better understand the complexity of computing defective and list defective colorings in general graphs. It is currently not even known if a $d$-defective coloring with $o((\Delta/d)^2)$ colors can be computed in $f(\Delta)\cdot \log^* n$ rounds for any function $f(\Delta)$.

\bibliographystyle{alpha}
\bibliography{references}
\end{document}